\documentclass[5p,times,numbers,sort&compress]{elsarticle}  
\usepackage{graphicx}
\usepackage{booktabs}
\usepackage{amsmath,amssymb}
\usepackage{newtxtext,newtxmath} 

\newtheorem{assumption}{Assumption}
\usepackage{pgfgantt}
\usepackage{hyperref}
\usepackage{mathrsfs}
\usepackage{subcaption}
\usepackage{textcomp}
\usepackage{xcolor}
\usepackage{enumitem}
\usepackage{multirow}

\newtheorem{theorem}{Theorem}
\newtheorem{proposition}{Proposition}
\newproof{pf}{Proof}

\usepackage{algorithm}        
\usepackage{algpseudocode}    
\usepackage[most]{tcolorbox}  
\usepackage{caption}          
\algrenewcommand\algorithmicrequire{\textbf{Inputs:}}
\algrenewcommand\algorithmicensure{\textbf{Output:}}
\usepackage{stfloats}

\journal{Applied Energy}   

\begin{document}

\begin{frontmatter}

\title{End-to-End Differentiable Predictive Control with Probabilistic Constraint-Satisfaction Guarantees for Building Demand Response}

\author[1]{Kaipeng Xu}
\author[1]{Zhuo Zhi}
\author[1]{Ruixuan Zhao}
\author[1]{Keyue Jiang}

\address[1]{Department of Electronic and Electrical Engineering, University College London, Torrington Place, London WC1E 7JE, United Kingdom}

\begin{abstract}
Building Demand Response (DR) benefits from Model Predictive Control (MPC) because thermal flexibility can be optimized over a future horizon while enforcing operational constraints, but the repeated online optimization required by MPC motivates Differentiable Predictive Control (DPC) as a computationally efficient explicit-policy alternative. Conventional DPC, however, separates model identification from policy optimization and remains sensitive to plant--model mismatch. This paper proposes an End-to-End DPC (E2E-DPC) framework that jointly optimizes the learned dynamics model and control policy using prediction, constraint, and economic objectives. Stage-wise conformal prediction (CP) calibrates recursive prediction errors to construct horizon-dependent tightened state constraints, while independent Hoeffding certification provides a high-confidence lower bound on plant-level satisfaction of the original state and input constraints under receding-horizon deployment. The framework is evaluated in a high-fidelity residential building simulation using EnergyPlus through the Energym interface and compared with online MPC and conventional DPC. Constraint tightening eliminates thermal violations for both DPC variants in the stress test and reduces violations by 99.4\% for online MPC. Compared with conventional DPC, E2E-DPC reduces the mean constraint-tightening margin at the first prediction stage by 31.0\% and increases the median certified plant-level lower bound from 0 to 0.814. The results further reveal a trade-off between stronger certification of plant-level constraint satisfaction and the ability to maintain nonempty tightened state sets as the prediction horizon increases, with stronger multi-step identification mitigating this trade-off.
\end{abstract}

\begin{keyword}
Differentiable Predictive Control, Building Demand Response, End-to-End Learning, Conformal Prediction, Constraint Tightening, Probabilistic Certification
\end{keyword}

\end{frontmatter}

\section{Introduction}
\label{sec:introduction}
Buildings represent one of the largest global energy consumers, accounting for approximately 40\% of total energy demand and over 30\% of CO\(_2\) emissions~\cite{hu2017survey,eu_energy2023}. Among various building loads, Heating, Ventilation, and Air Conditioning (HVAC) systems are particularly significant, consuming nearly 50\% of total building energy~\cite{center2021annual}. The substantial energy footprint and inherent thermal storage capacity of buildings therefore provide an important source of flexibility for Demand Response (DR) programs. As defined in~\cite{albadi2007demand}, DR involves changes in electricity usage by end-users in response to price signals. For buildings, this allows the thermal mass to act as a virtual battery; by pre-heating or pre-cooling during off-peak hours, energy consumption can be shifted away from high-price periods to reduce costs and support grid stability~\cite{ma2012demand,vedullapalli2019combined}.

Model Predictive Control (MPC) is well suited to HVAC-based DR because it can explicitly account for multivariable dynamics and operational constraints while anticipating future system behaviour over a finite prediction horizon~\cite{drgovna2020all}. Its effectiveness, however, depends strongly on the quality of the predictive model. While first-principles building models are often time-consuming to develop and require substantial domain expertise~\cite{moriyasu2021structured}, data-driven models based on machine learning provide a flexible alternative for representing complex whole-building thermal dynamics~\cite{gonzalez2017data,azuatalam2020reinforcement}. This flexibility introduces an important computational challenge: MPC repeatedly solves a finite-horizon optimization problem online, and embedding expressive nonlinear neural predictors can further increase the associated computational burden. Although tractability-oriented model architectures can substantially reduce the solution time of learning-based building MPC~\cite{xu2026input}, the finite-horizon optimization must still be solved anew at every control instant.

Differentiable Predictive Control (DPC) provides an alternative route to reducing this online burden. Rather than solving the predictive control problem online at every control instant, DPC formulates the closed-loop prediction and control problem as a differentiable computational graph and optimizes an explicit neural control policy offline~\cite{drgovna2022differentiable,drgovna2024learning}. In the conventional data-driven formulation for unknown nonlinear systems, a neural dynamics model is first identified from measured trajectories and subsequently held fixed while the control policy is optimized by backpropagating an MPC-inspired finite-horizon loss through the learned model~\cite{drgovna2022differentiable}. Once trained, the resulting policy can be evaluated directly during deployment, shifting the main computational burden from repeated online optimization to offline learning.

Despite these advantages, two limitations remain central to data-driven DPC. First, the conventional identify-then-optimize procedure decouples system identification from the downstream control objective. A model optimized for predictive accuracy is not necessarily the model that yields the best downstream control performance, an issue commonly referred to as \emph{objective mismatch}~\cite{lambert2020objective}. Once the dynamics model has been identified and frozen, gradients associated with constraint satisfaction and economic performance can no longer refine its representation for the control task. Second, because policy optimization relies on an approximate learned dynamics model, nominal constraint satisfaction on model-generated trajectories does not necessarily translate into constraint satisfaction on the physical plant. This limitation has already been observed in building applications of DPC under plant--model mismatch~\cite{drgovna2021building}.

Several subsequent studies have introduced mechanisms for improving the reliability of DPC policies~\cite{mukherjee2022neural,drgovna2024learning,cortez2022differentiable}. These methods provide important probabilistic, Lyapunov-based, or barrier-based guarantees, but typically rely on explicitly available nominal dynamics, additional model-dependent ingredients, or online safety intervention. Such assumptions are less natural in data-driven DPC based on identified nonlinear models, where residual plant--model prediction errors remain unavoidable and must be accounted for when enforcing constraints on the deployed policy.

To address these challenges, this paper develops an \textbf{End-to-End Differentiable Predictive Control (E2E-DPC) framework with probabilistic constraint-satisfaction guarantees}. The proposed framework jointly optimizes the learned dynamics model and explicit control policy so that prediction fidelity and downstream control performance can jointly shape the learned representation. An Encoder-Only Transformer (EoT) is employed as the differentiable dynamics model to represent the nonlinear temporal behaviour of the residential building. To account explicitly for the remaining mismatch between the learned predictor and the plant, recursive prediction errors are calibrated using stage-wise conformal prediction (CP) to construct horizon-dependent tightened state constraints. After model and policy selection, an independent Hoeffding certification stage provides a high-confidence lower bound on the plant-level probability of satisfying the original state and input constraints for the receding-horizon transition actually applied.

The main contributions of this work are threefold:
\begin{itemize}
    \item \textbf{Performance-oriented E2E-DPC:} A unified training paradigm is developed to jointly optimize the EoT dynamics model and the explicit DPC policy. One-step and recursive multi-step prediction losses provide complementary predictive anchors, while constraint and economic objectives allow downstream control performance to refine the learned dynamics representation. A two-stage warm start, adaptive gradient coupling, and Projecting Conflicting Gradients (PCGrad) are used to stabilize joint model--policy optimization while preserving a prediction-oriented anchor for the learned dynamics.

    \item \textbf{Data-driven probabilistic constraint certification under plant--model mismatch:} A stage-wise conformal constraint-tightening mechanism is developed directly from recursive plant--model prediction errors. The calibrated error radii act as constraint-tightening margins in the virtual rollout, while an independent Hoeffding certification stage yields a high-confidence plant-level lower bound on original state and input constraint satisfaction for the receding-horizon transition actually executed. The resulting framework does not require a pre-designed Lyapunov function, invariant terminal set, control barrier function, or auxiliary online safety optimization.

    \item \textbf{High-fidelity empirical evaluation:} The proposed framework is evaluated in a residential building DR problem using a high-fidelity EnergyPlus-based closed-loop simulation. Controlled comparisons among Online MPC, Conventional DPC, and E2E-DPC assess the closed-loop effect of constraint tightening, while an extensive hyperparameter study evaluates the stage-wise calibrated error radii and plant-level constraint-satisfaction certificates. The experiments further reveal a trade-off between stronger certification of plant-level constraint satisfaction and the ability to preserve nonempty tightened state sets as the prediction horizon increases, together with the role of multi-step identification in mitigating this trade-off.
\end{itemize}

The remainder of this paper is organized as follows. Section~\ref{sec:related_work} reviews DPC, performance-oriented and E2E learning, Transformer-based dynamics modeling, and data-driven approaches to constraint satisfaction under model uncertainty. Section~\ref{sec:e2e_dpc_framework} presents the proposed E2E-DPC formulation, loss design, and training procedure. Section~\ref{sec:constraint_tightening} develops the stage-wise conformal constraint-tightening and probabilistic certification framework. Section~\ref{sec:case_study} evaluates the proposed approach in the high-fidelity building DR case study. Finally, Section~\ref{sec:conclusion} summarizes the main findings and discusses limitations and future research directions.

\section{Related Work}
\label{sec:related_work}
\subsection{Differentiable Predictive Control}

MPC determines the control action by repeatedly solving a finite-horizon optimization problem online. Explicit Model Predictive Control (eMPC) provides an alternative implementation in which the parametric solution of the MPC problem is pre-computed offline, allowing the control action to be obtained by direct evaluation of an explicit control law during deployment. However, classical eMPC based on multiparametric programming can suffer from rapidly increasing solution complexity and memory requirements as the problem dimension grows~\cite{drgovna2022differentiable}. DPC provides a learning-based route to explicit predictive control by instead optimizing a neural control policy offline through a differentiable closed-loop model~\cite{drgovna2022differentiable,drgovna2024learning}. In the conventional data-driven formulation for unknown nonlinear systems, the dynamics model is first identified from measured trajectories using a multi-step prediction objective and is subsequently held fixed while the control policy is optimized by backpropagating an MPC-inspired finite-horizon loss through the learned model~\cite{drgovna2022differentiable}. This shifts the main computational burden from online optimization to offline learning and enables rapid policy evaluation during deployment. DPC has since been demonstrated across diverse applications, including multi-zone building control~\cite{drgovna2021building}, power-system economic dispatch~\cite{king2022koopman}, and large-scale urban traffic networks~\cite{tumu2024differentiable}.

Multi-step identification is well motivated for predictive control because it penalizes prediction errors accumulated over model rollouts. Nevertheless, the conventional identify-then-optimize structure introduces two limitations relevant to the present work. First, once system identification is completed, gradients associated with downstream control performance no longer refine the dynamics representation; a model optimized purely for identification accuracy is therefore not necessarily the most suitable representation for the subsequent control task. Second, the original nonlinear DPC formulation incorporates constraint-related penalties both in the multi-term system identification loss and again in the control-policy loss~\cite{drgovna2022differentiable}. Consequently, operational constraints influence two separately optimized stages, rather than being evaluated within a unified control-oriented training objective. Moreover, because the policy is optimized through the identified model rather than the physical plant, its closed-loop behavior remains sensitive to plant--model mismatch. This issue was explicitly identified in the multi-zone building study of~\cite{drgovna2021building}, where plant--model mismatch was reported as a limitation and left for future investigation. These observations motivate the performance-oriented E2E formulation discussed in the following subsection.

Several approaches have subsequently been proposed to improve the reliability of DPC policies. Mukherjee et al.~\cite{mukherjee2022neural} introduced Neural Lyapunov DPC, which jointly learns a control policy and a neural Lyapunov function and provides sampling-based probabilistic guarantees using Hoeffding's inequality. The formulation assumes an available and controllable nominal dynamics model, while extensions to learned identification models and plant--model mismatch were explicitly left as future work. Drgoňa et al.~\cite{drgovna2024learning} developed a broader probabilistic verification framework based on closed-loop indicator functions and Hoeffding's inequality; its stability analysis additionally assumes the existence of a local Lyapunov function, an invariant terminal set, and a feasible terminal controller. Cortez et al.~\cite{cortez2022differentiable} instead combine DPC with sampled-data control barrier functions to obtain robust safety guarantees, but require a prescribed barrier construction, model-dependent bounds, and an event-triggered online safety intervention. More recently, Viljoen et al.~\cite{viljoen2024differentiable} proposed a data-driven predictive safety filter that constructs a safe set from DPC training trajectories and modifies the nominal policy online when necessary; the authors note that the finite-data safe-set construction relies on assumptions that are not rigorous in general.

These developments provide complementary routes toward reliable DPC, but rely on additional structural assumptions, model-dependent ingredients, or online safety mechanisms. None of these approaches directly calibrates observed plant--model prediction errors of an identified dynamics model into stage-wise constraint tightening for DPC. The framework developed in this work addresses the two preceding limitations through separate but complementary components: E2E model--policy optimization removes the identify-then-optimize separation, while CP calibrates residual plant--model errors for constraint tightening and independent certification provides a lower bound on plant-level constraint satisfaction.

\subsection{Performance-Oriented and End-to-End Learning for Control}
\label{subsec:e2e_related_work}

A fundamental limitation of decoupled model-based control is the mismatch between the objective used for system identification and the objective of the downstream control task. Lambert et al.~\cite{lambert2020objective} formalized this issue as \emph{objective mismatch}: improving one-step model likelihood or prediction accuracy does not necessarily improve the performance obtained when the model is subsequently used for planning or control. This observation motivates task-oriented model learning, in which the downstream control objective contributes directly to the optimization of the learned dynamics representation.

Differentiable optimization provides a natural mechanism for such training. Amos et al.~\cite{amos2018differentiable} demonstrated that gradients can be propagated through an MPC controller to learn both dynamics and objective parameters in an E2E manner. More recently, task-oriented surrogate-model learning has been investigated directly in nonlinear MPC. Mayfrank et al.~\cite{mayfrank2024end} showed that surrogate models trained purely for system-identification accuracy can be suboptimal when embedded in economic nonlinear MPC, whereas E2E optimization of the surrogate model for the downstream control task can improve closed-loop performance. These results support replacing the conventional identify-then-optimize pipeline with a training paradigm in which prediction fidelity and finite-horizon control performance jointly shape the learned dynamics model.

Joint model--policy optimization, however, introduces a multi-objective optimization problem for the shared dynamics parameters~\cite{sener2018multi}. In such settings, gradients associated with different objectives may differ substantially in magnitude or conflict in direction, complicating naive gradient aggregation~\cite{chen2018gradnorm,yu2020gradient}. Existing multi-task learning approaches address this issue either by adapting loss weights, as in GradNorm~\cite{chen2018gradnorm} and uncertainty-based weighting~\cite{kendall2018multi}, or by directly modifying conflicting gradients. PCGrad~\cite{yu2020gradient} projects conflicting gradient components and provides a computationally simple gradient-surgery mechanism.

The framework developed in this work uses a two-stage warm-start and adaptive gradient-coupling strategy. The dynamics model is first trained using the prediction-oriented objective until its validation identification loss reaches a plateau, after which joint model--policy training is activated. During joint training, PCGrad~\cite{yu2020gradient} removes conflicting components of the finite-horizon gradient, while a discrete coupling coefficient regulates its contribution to the dynamics-model update. The coupling strength is adjusted according to validation performance to balance preservation of the predictive anchor with the influence of control-relevant information on the learned dynamics representation.
\subsection{Transformer-Based Dynamics Modeling for Building Control}
\label{subsec:transformer_related_work}

Accurate data-driven building control requires dynamics models capable of representing temporal dependencies and interactions among zone temperatures, HVAC actions, energy consumption, and exogenous conditions. Transformer architectures provide an alternative to recurrent sequence models by using self-attention to represent dependencies across observations while allowing parallel processing of the historical sequence~\cite{vaswani2017attention}. Recent studies have demonstrated the effectiveness of Transformer-based architectures for building cooling-load, indoor-temperature, and energy-consumption forecasting~\cite{li2023transformer,spencer2025transfer,gu2025large}. Encoder-based Transformer models have also been explored for building-energy prediction and sequential HVAC control~\cite{spencer2025transfer,kim2025realtime}, supporting the use of self-attention encoders for extracting control-relevant temporal representations from building operational data.

Most of these studies primarily consider forecasting or employ the learned model as a fixed predictive component. In the present work, an EoT instead serves as the differentiable dynamics model within DPC and is recursively evaluated over the prediction horizon while being jointly refined with the control policy. This requires an architecture that is expressive enough to represent nonlinear temporal dynamics while remaining compact and fully differentiable for repeated rollout and E2E optimization. The task-specific rationale for adopting an encoder-only rather than an encoder--decoder Transformer is detailed in Section~\ref{subsec:eot_dynamics}.

\subsection{Robust and Data-Driven Constraint Tightening under Model Uncertainty}
\label{subsec:robust_constraint_related_work}

Robust constraint satisfaction is a central challenge in MPC when the predictive model is affected by disturbances or model mismatch, and becomes particularly important in learning-based MPC where part or all of the dynamics may be identified from data~\cite{hewing2020learning}. Classical robust formulations include min--max MPC, which explicitly optimizes against worst-case uncertainty but can lead to computationally demanding nonlinear optimization problems~\cite{raimondo2009min}. Tube-based MPC provides a more tractable alternative by restricting a nominal trajectory to tightened constraints such that the uncertain trajectory remains within the original admissible set~\cite{mayne2005robust}. While constructive tube and constraint-tightening procedures are well established for linear systems~\cite{mayne2005robust,chisci2001systems}, their extension to general nonlinear dynamics is substantially more difficult and can require restrictive system-theoretic ingredients or lead to conservative uncertainty propagation~\cite{marruedo2002input,mayne2011tube,kohler2018novel}.

Two representative nonlinear formulations illustrate this difficulty. Marruedo et al.~\cite{marruedo2002input} model plant--model mismatch as bounded additive uncertainty and propagate its effect using a Lipschitz constant $L_f$ of the nominal dynamics. The resulting $j$-step prediction-error bound scales as $\sum_{i=0}^{j-1}L_f^i$, and the corresponding uncertainty set is subtracted from the nominal state constraints. The construction is transparent, but the authors note that global Lipschitz propagation can be overly conservative; their robust MPC formulation additionally employs a terminal cost, terminal constraint, local controller, and Lyapunov-based terminal ingredients. Köhler et al.~\cite{kohler2018novel} avoid the unbounded growth of this Lipschitz-based tightening by exploiting local incremental stabilizability and an achievable contraction rate $\rho<1$, resulting in a bounded growing tube. Although the resulting tightening itself only requires the decay-rate estimate, establishing the required incremental stabilizability property and a suitable contraction bound remains a nontrivial model-dependent task for general nonlinear systems.

These difficulties are amplified for learned sequence models. Seel et al.~\cite{seel2021neural} show that a neural nonlinear autoregressive model with exogenous inputs predictor depending on past inputs and outputs can be represented as a state-space model by augmenting the state with its historical variables. The EoT dynamics model considered in this work requires an analogous history augmentation. Under a standard unweighted norm, the resulting augmented transition contains a shift operation that directly copies part of the historical state to the next augmented state and is therefore not strictly contractive in those directions. Consequently, a Lipschitz-based error bound of the form used in~\cite{marruedo2002input} grows at least linearly with the prediction stage when $L_f=1$ and geometrically when $L_f>1$, which can rapidly reduce or even eliminate the tightened feasible region. Weighted norms or local bounds can reduce this conservatism~\cite{marruedo2002input}, but obtaining sufficiently tight analytical bounds for a high-dimensional learned EoT model remains difficult. Similarly, applying the incremental-stabilizability approach of~\cite{kohler2018novel} would require establishing a suitable contraction property for the history-augmented learned dynamics, which is not naturally available during E2E model--policy learning.

These considerations motivate a data-driven alternative that calibrates prediction uncertainty directly from observed model errors rather than propagating a global worst-case sensitivity bound. CP provides model-agnostic prediction regions with finite-sample marginal coverage under exchangeability and can be applied on top of arbitrary regression models~\cite{lei2018distribution}. Of particular relevance, Chee et al.~\cite{chee2024uncertainty} combine conformal uncertainty quantification with dynamic constraint tightening to robustify learned model-based controllers. Their framework computes an uncertainty vector online and applies the same current uncertainty width across the reference-prediction horizon; the estimated uncertainty is deliberately not propagated with the look-ahead stage in order to promote feasibility. Inspired by this data-calibrated tightening principle, the present work instead calibrates a separate constraint-tightening margin for each prediction stage using recursive multi-step rollout errors. This directly represents the horizon-dependent error accumulation of the learned EoT model and integrates the resulting tightening into E2E-DPC training, while avoiding explicit Lipschitz or incremental-stability calculations and additional online reference-generation optimization.

\section{The End-to-End Differentiable Predictive Control Framework}
\label{sec:e2e_dpc_framework}

This section details the proposed E2E-DPC framework. The framework is designed to jointly learn the system dynamics and a control policy through E2E training. It builds upon the structure of DPC~\cite{drgovna2022differentiable,drgovna2024learning}, while introducing a performance-oriented unified training paradigm that addresses limitations of conventional decoupled DPC training. A stage-dependent conformal constraint-tightening mechanism and its probabilistic certification are introduced in Section~\ref{sec:constraint_tightening}.

\subsection{General Problem Formulation}
\label{sec:pocp_formulation}

The E2E-DPC framework addresses a parametric optimal control problem (pOCP), where the objective is to learn the parameters of both the dynamics model and the control policy simultaneously to minimize an expected cumulative cost. Let $N_p$ denote the prediction horizon, define the local prediction-stage index set as $\mathcal{K}:=\{0,\ldots,N_p-1\}$, and let $\mathbf{U}:=\{\mathbf{u}_k\}_{k\in\mathcal{K}}$ denote the control sequence over the horizon. At each receding-horizon decision instant, the local stage index is reset such that $k=0$ denotes the current decision stage. Let $\mathbf{y}_0\in\mathbb{R}^{n_y}$ denote the current measured model-variable vector, $\hat{\mathbf{y}}_k\in\mathbb{R}^{n_y}$ its recursively predicted continuation, and $\hat{\mathbf{x}}_k\in\mathbb{R}^{n_x}$ the subvector of $\hat{\mathbf{y}}_k$ subject to operational state constraints. The corresponding measured and recursively generated historical input windows are denoted by $\mathcal{H}_0$ and $\widehat{\mathcal{H}}_k$, respectively. The cost-parameter sequence and horizon-level pOCP parameter vector are defined as
\begin{equation}
\begin{aligned}
\boldsymbol{\pi}&:=[\pi_0,\ldots,\pi_{N_p-1}]^\top,\\
\boldsymbol{\eta}&:=[\boldsymbol{\pi}^\top,\mathbf{x}_\ell^\top,\mathbf{x}_h^\top,\mathbf{u}_\ell^\top,\mathbf{u}_h^\top]^\top\in\mathbb{R}^{n_\eta}.
\end{aligned}
\label{eq:pocp_parameters}
\end{equation}

Unless otherwise stated, all stage-wise relations below are imposed for every $k\in\mathcal{K}$. The general formulation is as follows:

\begin{subequations}
\label{eq:pocp}
\begin{align}
\min_{\theta_x,\theta_u}\quad
&\mathbb{E}_{\mathcal{H}_0\sim\mathcal{P}_{\mathcal{H}_0},\,\boldsymbol{\eta}\sim\mathcal{P}_{\eta}}
\left[
\sum_{k\in\mathcal{K}}
\ell(\hat{\mathbf{y}}_{k+1},\mathbf{u}_k,\pi_k)
\right]
\label{eq:pocp_objective}\\
\mathrm{s.t.}\quad
&\hat{\mathbf{y}}_{k+1}
=f_{\theta_x}(\widehat{\mathcal{H}}_k,\mathbf{u}_k),
\label{eq:pocp_dynamics}\\
&\widehat{\mathcal{H}}_{k+1}
=\operatorname{shift}\!\left(
\widehat{\mathcal{H}}_k,
\hat{\mathbf{y}}_{k+1},
\mathbf{u}_k
\right),
\label{eq:pocp_history_update}\\
&\mathbf{U}=f_{\theta_u}(\mathbf{y}_0,\boldsymbol{\eta}),
\label{eq:pocp_policy}\\
&\hat{\mathbf{x}}_{k+1}\in
\mathcal{X}:=
\{\mathbf{x}\in\mathbb{R}^{n_x}\mid
\mathbf{x}_\ell\leq\mathbf{x}\leq\mathbf{x}_h\},
\label{eq:pocp_state_constraints}\\
&\mathbf{u}_k\in
\mathcal{U}:=
\{\mathbf{u}\in\mathbb{R}^{n_u}\mid
\mathbf{u}_\ell\leq\mathbf{u}\leq\mathbf{u}_h\},
\label{eq:pocp_input_constraints}\\
&\widehat{\mathcal{H}}_0=\mathcal{H}_0,\qquad
\hat{\mathbf{y}}_0=\mathbf{y}_0.
\label{eq:pocp_initial_condition}
\end{align}
\end{subequations}

The objective in~\eqref{eq:pocp_objective} minimizes the expected cumulative stage cost over a finite prediction horizon of length $N_p$. The expectation is taken over the measured initial history, $\mathcal{P}_{\mathcal{H}_0}$, and the horizon-level pOCP parameter vector, $\mathcal{P}_{\eta}$. At prediction stage $k$, the control input $\mathbf{u}_k$ induces the one-step-ahead model-variable prediction $\hat{\mathbf{y}}_{k+1}$, and $\pi_k$ denotes the cost parameter associated with this transition. The optimization is performed jointly over the dynamics-model parameters $\theta_x$ and the policy parameters $\theta_u$.

The dynamics and history-update mappings are defined in~\eqref{eq:pocp_dynamics} and~\eqref{eq:pocp_history_update}. The notation $f_{\theta_x}(\widehat{\mathcal{H}}_k,\mathbf{u}_k)$ represents evaluation of the history-based predictor after incorporating the current control input into the current model-input window. The resulting complete model-variable prediction $\hat{\mathbf{y}}_{k+1}$ is then written back together with $\mathbf{u}_k$ to construct the next recursive window $\widehat{\mathcal{H}}_{k+1}$. The constrained state $\hat{\mathbf{x}}_{k+1}$ is the designated subvector of $\hat{\mathbf{y}}_{k+1}$ on which the state constraints are imposed. The explicit policy $f_{\theta_u}$ maps the current measured model-variable vector $\mathbf{y}_0$ and the horizon-level pOCP parameters to the complete control sequence $\mathbf{U}$. The vector $\boldsymbol{\eta}$ in~\eqref{eq:pocp_parameters} collects problem information that may condition the policy; parameters that remain fixed for a particular application need not be supplied explicitly as network inputs.

The present pOCP uses a finite prediction horizon without an explicit terminal constraint set or terminal penalty. Uncertainty is handled through the stage-dependent conformal tightening sequence developed in Section~\ref{sec:constraint_tightening}. The complete tightening sequence is used throughout the virtual prediction horizon for policy synthesis, whereas the probabilistic plant-level certificate is formulated for the first receding-horizon transition that is actually executed.

\subsection{Differentiable Closed-Loop System}

The core of the E2E-DPC framework is a differentiable representation of the closed-loop system, enabling gradient-based optimization of the dynamics model and control policy through the unrolled prediction horizon. The framework consists of two primary learnable components: the system dynamics model $f_{\theta_x}$ and the control policy $f_{\theta_u}$.

\subsubsection{System Dynamics Model}
\label{subsec:eot_dynamics}

The dynamics model $f_{\theta_x}$ predicts the next complete model-variable vector from a fixed-length historical window and the current control input, as represented in~\eqref{eq:pocp_dynamics}. Based on the suitability of self-attention for multivariate building time-series modeling discussed in Section~\ref{subsec:transformer_related_work}, an EoT is adopted as the differentiable dynamics model.

The encoder-only architecture is aligned with the prediction task considered here. The original Transformer employs an encoder--decoder structure for sequence-to-sequence generation~\cite{vaswani2017attention}, whereas $f_{\theta_x}$ performs a many-to-one regression from a historical window of length $N_h$ to the next-step model-variable vector. The self-attention encoder therefore constructs a context-aware representation of the historical sequence, which is mapped directly to the one-step prediction through a regression head; a separate autoregressive decoder is not required. Multi-step predictions are instead obtained by recursively applying the same one-step model and updating the historical window according to~\eqref{eq:pocp_dynamics}--\eqref{eq:pocp_history_update}. This architecture provides a compact differentiable module for repeated finite-horizon rollout while allowing the dynamics representation to be jointly refined with the control policy.

\subsubsection{Control Policy}

The control policy $f_{\theta_u}$ is an explicit neural network that maps the currently available system information and the pOCP parameters to the complete future control sequence $\mathbf{U}$ defined in Section~\ref{sec:pocp_formulation}. In the implementation considered in this work, the policy is evaluated once at the beginning of each virtual rollout. An $L$-layer feed-forward implementation is written as:

\begin{subequations}
\label{eq:policy_network}
\begin{align}
\mathbf{U}&=f_{\theta_u}(\mathbf{y}_0,\boldsymbol{\eta})=H_L\mathbf{z}_L+\mathbf{b}_L,\\
\mathbf{z}_l&=\sigma(H_{l-1}\mathbf{z}_{l-1}+\mathbf{b}_{l-1}),\qquad l=1,\ldots,L,\\
\mathbf{z}_0&=[\mathbf{y}_0^\top,\boldsymbol{\eta}^\top]^\top.
\end{align}
\end{subequations}

where $\mathbf{z}_l$ denotes the hidden-layer activation, $\sigma$ is an element-wise activation function, and $\theta_u=\{H_l,\mathbf{b}_l\}_{l=0}^{L}$ collects the learnable weights and biases. The output layer has dimension $N_p\times n_u$ and is reshaped into the $N_p$-stage control sequence $\mathbf{U}$. 

Together, the dynamics model and the policy form a differentiable closed-loop system where the model-variable trajectory is recursively generated as:
\begin{subequations}
\label{eq:virtual_rollout}
\begin{align}
\hat{\mathbf{y}}_{k+1}
&=f_{\theta_x}(\widehat{\mathcal{H}}_k,\mathbf{u}_k),
\qquad k\in\mathcal{K},\\
\widehat{\mathcal{H}}_{k+1}
&=\operatorname{shift}\!\left(
\widehat{\mathcal{H}}_k,
\hat{\mathbf{y}}_{k+1},
\mathbf{u}_k
\right),\\
\widehat{\mathcal{H}}_0&=\mathcal{H}_0,\qquad
\hat{\mathbf{y}}_0=\mathbf{y}_0.
\end{align}
\end{subequations}

The policy is evaluated once at the beginning of the prediction horizon to generate the complete control sequence $\mathbf{U}$. The successive elements of this sequence are then incorporated into the recursive historical windows and propagated through $f_{\theta_x}$ to generate the $N_p$-step virtual model-variable trajectory. At each stage, the constrained-state vector $\hat{\mathbf{x}}_k$ is obtained as the corresponding constrained subvector of $\hat{\mathbf{y}}_k$. The differentiability of this computational graph with respect to both $\theta_x$ and $\theta_u$ enables the finite-horizon control loss to update the dynamics model and the control policy jointly.

\subsection{Performance-Oriented Loss Design}
To solve the pOCP in~\eqref{eq:pocp}, its objective and constraints are represented by a composite differentiable loss. Four complementary terms are used: a one-step identification loss, a multi-step rollout identification loss, a closed-loop constraint-violation loss, and a finite-horizon control objective. Let $n_b$ denote the mini-batch size. The overall objective is
\begin{equation}
\mathcal{L}_{\mathrm{e2e}}=\lambda_{\mathrm{id}}\mathcal{L}_{\mathrm{id}}+\mathcal{L}_{\mathrm{dpc}}.
\label{eq:e2e_loss}
\end{equation}
\begin{equation}
\mathcal{L}_{\mathrm{dpc}}=\lambda_{\mathrm{ms}}\mathcal{L}_{\mathrm{ms}}+\lambda_{\mathrm{cons}}\mathcal{L}_{\mathrm{cons}}+\lambda_{\mathrm{obj}}\mathcal{L}_{\mathrm{obj}}.
\label{eq:dpc_loss}
\end{equation}

\begin{itemize}
\item \textbf{One-Step Identification Loss ($\mathcal{L}_{\mathrm{id}}$):} The one-step identification loss anchors the local predictive fidelity of $f_{\theta_x}$ using measured history--target pairs:
\begin{equation}
\mathcal{L}_{\mathrm{id}}
=
\frac{1}{n_b}
\sum_{i=1}^{n_b}
\left\|
\mathbf{y}_{j+1}^{(i)}
-
f_{\theta_x}\!\left(
\mathcal{H}_{j}^{(i)},
\mathbf{u}_{j,\mathrm{d}}^{(i)}
\right)
\right\|_2^2.
\label{eq:id_loss}
\end{equation}
where $j$ indexes the sampled one-step transition and $\mathbf{u}_{j,\mathrm{d}}^{(i)}$ denotes the corresponding recorded control input.

\item \textbf{Multi-Step Rollout Identification Loss ($\mathcal{L}_{\mathrm{ms}}$):} During joint training, the dynamics model is additionally evaluated in a free-running $N_p$-step rollout under the recorded future controls. For sample $i$,
\begin{subequations}
\label{eq:ms_rollout}
\begin{align}
\widetilde{\mathcal{H}}_0^{(i)}
&=\mathcal{H}_0^{(i)},\\
\widetilde{\mathbf{y}}_{k+1}^{(i)}
&=
f_{\theta_x}\!\left(
\widetilde{\mathcal{H}}_k^{(i)},
\mathbf{u}_{k,\mathrm{d}}^{(i)}
\right),
\qquad k\in\mathcal{K},\\
\widetilde{\mathcal{H}}_{k+1}^{(i)}
&=
\operatorname{shift}\!\left(
\widetilde{\mathcal{H}}_k^{(i)},
\widetilde{\mathbf{y}}_{k+1}^{(i)},
\mathbf{u}_{k,\mathrm{d}}^{(i)}
\right).
\end{align}
\end{subequations}
where $\mathbf{u}_{k,\mathrm{d}}^{(i)}$ denotes the recorded control applied in the data. The multi-step loss is
\begin{equation}
\mathcal{L}_{\mathrm{ms}}
=
\frac{1}{n_bN_p}
\sum_{i=1}^{n_b}
\sum_{k\in\mathcal{K}}
\left\|
\mathbf{y}_{k+1}^{(i)}
-
\widetilde{\mathbf{y}}_{k+1}^{(i)}
\right\|_2^2.
\label{eq:ms_loss}
\end{equation}
\item \textbf{Constraint Violation Loss ($\mathcal{L}_{\mathrm{cons}}$):} The state and input constraints in~\eqref{eq:pocp_state_constraints} and~\eqref{eq:pocp_input_constraints} are incorporated through squared-ReLU penalties evaluated on the policy-generated virtual rollout:
\begin{equation}
\begin{split}
\mathcal{L}_{\mathrm{cons}}=\frac{1}{n_bN_p}\sum_{i=1}^{n_b}\sum_{k\in\mathcal{K}}\Big(&\|\operatorname{ReLU}(\hat{\mathbf{x}}_{k+1}^{(i)}-\mathbf{x}_h)\|_2^2\\
&+\|\operatorname{ReLU}(\mathbf{x}_\ell-\hat{\mathbf{x}}_{k+1}^{(i)})\|_2^2\\
&+\|\operatorname{ReLU}(\mathbf{u}_k^{(i)}-\mathbf{u}_h)\|_2^2\\
&+\|\operatorname{ReLU}(\mathbf{u}_\ell-\mathbf{u}_k^{(i)})\|_2^2\Big).
\end{split}
\label{eq:constraint_loss}
\end{equation}
\item \textbf{Objective Loss ($\mathcal{L}_{\mathrm{obj}}$):} The objective loss is the mini-batch approximation of the expected finite-horizon cost in~\eqref{eq:pocp_objective}:
\begin{equation}
\mathcal{L}_{\mathrm{obj}}
=
\frac{1}{n_b}
\sum_{i=1}^{n_b}
\sum_{k\in\mathcal{K}}
\ell\!\left(
\hat{\mathbf{y}}_{k+1}^{(i)},
\mathbf{u}_k^{(i)},
\pi_k^{(i)}
\right).
\label{eq:objective_loss}
\end{equation}
\end{itemize}

The two prediction-fidelity terms play complementary roles. $\mathcal{L}_{\mathrm{id}}$ preserves local one-step predictive fidelity, whereas $\mathcal{L}_{\mathrm{ms}}$ penalizes the accumulation of prediction errors during recursive finite-horizon rollouts. The latter is particularly relevant to predictive control, since errors accumulated over the prediction horizon directly affect the virtual trajectories used for control-policy optimization. This trajectory-level objective is also consistent with the original DPC formulation in~\cite{drgovna2022differentiable}, where the neural state-space model is identified by comparing measured and predicted output trajectories over multiple prediction steps.

A central design choice of the proposed framework is to separate prediction fidelity from operational penalties and the task-level control objective. In the nonlinear DPC formulation of~\cite{drgovna2022differentiable}, constraint-related penalties are included both in the multi-term system-identification loss and in the subsequent control-policy loss. In the present formulation, $\mathcal{L}_{\mathrm{id}}$ and $\mathcal{L}_{\mathrm{ms}}$ do not directly contain operational constraint penalties, while constraint handling and the task-level control objective are represented by $\mathcal{L}_{\mathrm{cons}}$ and $\mathcal{L}_{\mathrm{obj}}$ on the policy-generated virtual rollout. This avoids duplicating operational constraint penalties across model identification and policy optimization and allows their influence to be controlled explicitly through the corresponding loss weights.

This separation does not decouple model and policy learning. During the joint stage, the control-relevant multi-step prediction term $\mathcal{L}_{\mathrm{ms}}$ and the policy-rollout terms $\mathcal{L}_{\mathrm{cons}}$ and $\mathcal{L}_{\mathrm{obj}}$ are combined in $\mathcal{L}_{\mathrm{dpc}}$ and jointly influence the dynamics model. Meanwhile, $\mathcal{L}_{\mathrm{cons}}$ and $\mathcal{L}_{\mathrm{obj}}$ additionally update the control policy through the differentiable virtual rollout. The learned dynamics representation can therefore be refined according to both predictive fidelity and its effect on finite-horizon control performance. The effect of this loss design and joint model--policy optimization is evaluated in Section~\ref{sec:case_study}.

\subsection{End-to-End Training Procedure}
\label{subsec:e2e_training}

The dynamics model and the control policy are optimized offline using a two-stage training procedure, as summarized in Figure~\ref{fig:e2e_training_diagram}. The first stage trains only the dynamics model until its validation identification loss satisfies a plateau criterion. The second stage activates joint E2E training, in which both models are updated using finite-horizon differentiable rollouts. This criterion-based warm start provides the dynamics model with a predictive initialization before the finite-horizon DPC gradient is introduced into its parameter updates, thereby reducing abrupt gradient interference at the transition to joint training. The duration of the first stage is therefore determined by validation performance rather than by a prescribed number of training epochs.

The procedure begins with the \textbf{Initial Dynamics Model Training} phase, as depicted in the upper panel of Figure~\ref{fig:e2e_training_diagram}. During this stage, the control policy is inactive and only the dynamics model $f_{\theta_x}$ is trained on measured one-step transitions. As illustrated by the green arrows, forward propagation produces a one-step-ahead prediction, which is compared with the measured next-step model-variable vector to compute the identification loss $\mathcal{L}_{\mathrm{id}}$ in~\eqref{eq:id_loss}. The resulting gradient is backpropagated through the dynamics model, as indicated by the red arrow, to update only $\theta_x$.

After each training epoch, $\mathcal{L}_{\mathrm{id}}$ is evaluated on the validation set and appended to the validation-loss history. Once at least $N_{\mathrm{pl}}$ validation values have been recorded, the relative improvement between the oldest and newest values in the latest $N_{\mathrm{pl}}$ records is evaluated. Joint training is activated from the following epoch when this relative improvement falls below the plateau threshold $\varepsilon_{\mathrm{pl}}$.

\begin{figure*}[t]
    \centering
    \includegraphics[width=\textwidth]{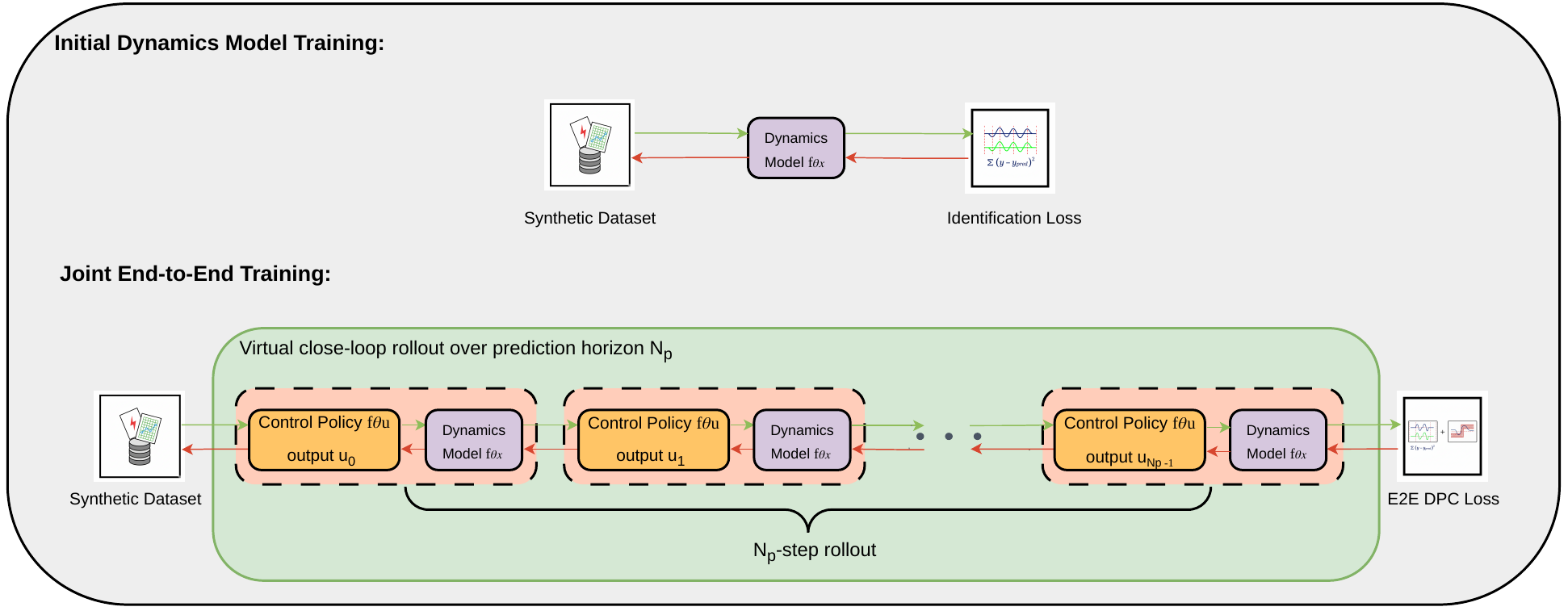}
    \caption{The two-stage E2E-DPC training procedure. Green arrows denote forward information flow used to construct predictions and evaluate the corresponding losses, whereas red arrows denote backward gradient propagation used to update the learnable parameters. In the upper panel, only the dynamics model $f_{\theta_x}$ is trained using the one-step identification loss $\mathcal{L}_{\mathrm{id}}$, while the control policy remains inactive. Once the validation identification loss satisfies the plateau criterion, joint E2E training is activated. In the lower panel, $f_{\theta_x}$ is recursively evaluated under recorded future controls to compute the multi-step loss $\mathcal{L}_{\mathrm{ms}}$, while $f_{\theta_u}$ generates the control sequence used in the policy-based virtual rollout for evaluating $\mathcal{L}_{\mathrm{cons}}$ and $\mathcal{L}_{\mathrm{obj}}$. These finite-horizon terms form $\mathcal{L}_{\mathrm{dpc}}$, whose gradient contributes to the dynamics-model update through PCGrad and the coupling coefficient $\beta$; its policy-dependent components additionally update $f_{\theta_u}$.}
\label{fig:e2e_training_diagram}
\end{figure*}

Once the validation identification loss satisfies the plateau criterion, the \textbf{Joint E2E Training} phase is activated, as illustrated in the lower panel of Figure~\ref{fig:e2e_training_diagram}. Two complementary finite-horizon rollouts are evaluated during each mini-batch. First, $f_{\theta_x}$ is recursively propagated under the recorded future control sequence to compute the multi-step identification loss $\mathcal{L}_{\mathrm{ms}}$ in~\eqref{eq:ms_loss}. Second, the policy $f_{\theta_u}$ is evaluated using the current measured model-variable vector $\mathbf{y}_0$ and the horizon-level pOCP parameters to generate the complete control sequence $\mathbf{U}$. Its successive elements are propagated through $f_{\theta_x}$ according to~\eqref{eq:virtual_rollout} to construct the policy-generated virtual trajectory and evaluate $\mathcal{L}_{\mathrm{cons}}$ and $\mathcal{L}_{\mathrm{obj}}$. Together, $\mathcal{L}_{\mathrm{ms}}$, $\mathcal{L}_{\mathrm{cons}}$, and $\mathcal{L}_{\mathrm{obj}}$ form the finite-horizon DPC loss $\mathcal{L}_{\mathrm{dpc}}$ defined in~\eqref{eq:dpc_loss}.

Because $\theta_x$ influences both one-step prediction and the finite-horizon rollouts, it receives two potentially conflicting gradient signals. Define
\begin{equation}
g_{\mathrm{id}}:=\nabla_{\theta_x}\!\left(\lambda_{\mathrm{id}}\mathcal{L}_{\mathrm{id}}\right),\qquad g_{\mathrm{dpc}}:=\nabla_{\theta_x}\mathcal{L}_{\mathrm{dpc}}.
\label{eq:task_gradients}
\end{equation}
Rather than directly summing these gradients, an asymmetric form of PCGrad~\cite{yu2020gradient} is used, with the identification gradient treated as the reference direction. The DPC gradient is projected only when the two gradients conflict:
\begin{equation}
\widetilde{g}_{\mathrm{dpc}}=g_{\mathrm{dpc}}-\frac{\min\{0,\langle g_{\mathrm{dpc}},g_{\mathrm{id}}\rangle\}}{\|g_{\mathrm{id}}\|_2^2+\varepsilon_{\mathrm{pc}}}g_{\mathrm{id}}.
\label{eq:pcgrad}
\end{equation}
where $\varepsilon_{\mathrm{pc}}>0$ is a small numerical constant. The resulting dynamics-model gradient is
\begin{equation}
g_x=g_{\mathrm{id}}+\beta\widetilde{g}_{\mathrm{dpc}}.
\label{eq:dynamics_update_gradient}
\end{equation}
The policy is updated using
\begin{equation}
g_u=\nabla_{\theta_u}\mathcal{L}_{\mathrm{dpc}}.
\label{eq:policy_update_gradient}
\end{equation}

Although $\mathcal{L}_{\mathrm{ms}}$ is control-relevant because it improves the recursive prediction behavior used in finite-horizon control, it is evaluated under recorded future controls and therefore has no direct dependence on $\theta_u$. Consequently, $\mathcal{L}_{\mathrm{ms}}$ contributes to the dynamics-model update but has zero direct contribution to $g_u$. In contrast, $\mathcal{L}_{\mathrm{cons}}$ and $\mathcal{L}_{\mathrm{obj}}$ depend on the policy-generated rollout and therefore update both the policy and, through the differentiable dynamics model, $f_{\theta_x}$. PCGrad and the coupling coefficient $\beta$ regulate the contribution of the complete finite-horizon DPC gradient to the dynamics-model update.

The coupling coefficient is selected from the discrete levels $\{\beta_1,\beta_2,\beta_3\}$. Joint training begins at the smallest level $\beta_1$, limiting the initial influence of the finite-horizon DPC gradient on the dynamics model. A transition to the next level is permitted only after the current level has been retained for a prescribed number of epochs, the validation identification loss has recovered sufficiently from its worst value observed after joint training was activated, and the composite validation objective satisfies the selected improvement criterion. If the composite validation objective subsequently fails to improve for a prescribed period, $\beta$ may be reduced by one level. This adaptive coupling schedule balances preservation of one-step predictive fidelity against the influence of the finite-horizon DPC objective on the learned dynamics representation. The complete two-stage procedure is summarized in Algorithm~\ref{alg:e2e_training}.

\refstepcounter{algorithm}\label{alg:e2e_training}
\begin{tcolorbox}[breakable,enhanced,
  title={Algorithm~\thealgorithm: E2E-DPC Training Procedure}]

{\small 
\begin{algorithmic}[1]

\Require Initial parameters $\theta_x$ and $\theta_u$; training and validation datasets; optimizers $\mathcal{O}_x$ and $\mathcal{O}_u$; plateau-detection window length $N_{\mathrm{pl}}$ and relative-improvement threshold $\varepsilon_{\mathrm{pl}}$; coupling levels $\{\beta_1,\beta_2,\beta_3\}$

\State \textbf{Stage 1: Initial dynamics-model training}
\State Initialize the validation identification-loss history $\mathcal{V}_{\mathrm{id}}\gets\varnothing$ and the joint-training flag $\mathrm{unlocked}\gets\mathrm{false}$
\While{$\mathrm{unlocked}=\mathrm{false}$}
\For{each mini-batch of measured one-step transitions}
\State Compute $\mathcal{L}_{\mathrm{id}}$ using~\eqref{eq:id_loss}
\State Update $\theta_x$ using $\nabla_{\theta_x}(\lambda_{\mathrm{id}}\mathcal{L}_{\mathrm{id}})$
\EndFor
\State Evaluate the validation identification loss and append it to $\mathcal{V}_{\mathrm{id}}$
\If{$\mathcal{V}_{\mathrm{id}}$ contains at least $N_{\mathrm{pl}}$ values and the relative improvement between the oldest and newest values in its latest $N_{\mathrm{pl}}$ records is below $\varepsilon_{\mathrm{pl}}$}
\State Set $\mathrm{unlocked}\gets\mathrm{true}$
\EndIf
\EndWhile

\Statex
\State \textbf{Stage 2: Joint E2E training}
\State Initialize $\beta\gets\beta_1$ and activate joint updates of $\theta_x$ and $\theta_u$
\For{each joint-training epoch}
\For{each mini-batch}
\State Sample initial history windows, measured one-step targets, recorded future model-variable--control trajectories, and horizon-level pOCP parameters
\State Compute $\mathcal{L}_{\mathrm{id}}$
\State Recursively roll out $f_{\theta_x}$ under the recorded future controls and compute $\mathcal{L}_{\mathrm{ms}}$ using~\eqref{eq:ms_loss}
\State $\mathbf{U}^{(i)}\gets f_{\theta_u}(\mathbf{y}_0^{(i)},\boldsymbol{\eta}^{(i)})$ for $i=1,\ldots,n_b$
\State Set $\widehat{\mathcal{H}}_0^{(i)}\gets\mathcal{H}_0^{(i)}$ and $\hat{\mathbf{y}}_0^{(i)}\gets\mathbf{y}_0^{(i)}$
\For{$k\in\mathcal{K}$}
\State Extract $\mathbf{u}_k^{(i)}$ from $\mathbf{U}^{(i)}$
\State $\hat{\mathbf{y}}_{k+1}^{(i)}\gets f_{\theta_x}(\widehat{\mathcal{H}}_k^{(i)},\mathbf{u}_k^{(i)})$
\State Extract $\hat{\mathbf{x}}_{k+1}^{(i)}$ from $\hat{\mathbf{y}}_{k+1}^{(i)}$
\State $\widehat{\mathcal{H}}_{k+1}^{(i)}\gets\operatorname{shift}(\widehat{\mathcal{H}}_k^{(i)},\hat{\mathbf{y}}_{k+1}^{(i)},\mathbf{u}_k^{(i)})$
\EndFor
\State Compute $\mathcal{L}_{\mathrm{cons}}$ and $\mathcal{L}_{\mathrm{obj}}$
\State $\mathcal{L}_{\mathrm{dpc}}\gets\lambda_{\mathrm{ms}}\mathcal{L}_{\mathrm{ms}}+\lambda_{\mathrm{cons}}\mathcal{L}_{\mathrm{cons}}+\lambda_{\mathrm{obj}}\mathcal{L}_{\mathrm{obj}}$
\State $g_{\mathrm{id}}\gets\nabla_{\theta_x}(\lambda_{\mathrm{id}}\mathcal{L}_{\mathrm{id}})$ and $g_{\mathrm{dpc}}\gets\nabla_{\theta_x}\mathcal{L}_{\mathrm{dpc}}$
\State Compute $\widetilde{g}_{\mathrm{dpc}}$ using~\eqref{eq:pcgrad}
\State $g_x\gets g_{\mathrm{id}}+\beta\widetilde{g}_{\mathrm{dpc}}$
\State $g_u\gets\nabla_{\theta_u}\mathcal{L}_{\mathrm{dpc}}$
\State Update $\theta_x$ using $\mathcal{O}_x$ and $g_x$, and update $\theta_u$ using $\mathcal{O}_u$ and $g_u$
\EndFor
\State Evaluate the validation identification loss and the composite validation objective
\State Update $\beta$ according to the identification-recovery, minimum-hold, composite-improvement, and reversion criteria
\EndFor

\Ensure Optimized dynamics model $f_{\theta_x}$ and control policy $f_{\theta_u}$

\end{algorithmic}
} 
\end{tcolorbox}

\section{Constraint Tightening and Probabilistic Certification}
\label{sec:constraint_tightening}

This section develops the probabilistic constraint-satisfaction mechanism for the proposed E2E-DPC framework. Stage-wise conformal calibration is used to construct a horizon-dependent tightening sequence from recursive model-rollout errors, and Hoeffding's inequality is subsequently used to certify nominal satisfaction of the tightened constraints. These two components are then combined to obtain the plant-level probabilistic constraint-satisfaction guarantee.

\subsection{Recursive Rollout Error Quantification}

To quantify the prediction uncertainty accumulated over the finite horizon, the learned dynamics model is evaluated recursively under recorded control sequences. This construction reproduces the recursive prediction mechanism used by the virtual rollout and yields a stage-dependent prediction error for each look-ahead step.

For sample $i$, let $\mathcal{H}_0^{(i)}$ denote the measured initial history window, let $\{\mathbf{u}_{k,\mathrm{d}}^{(i)}\}_{k\in\mathcal{K}}$ denote the recorded future control sequence, and let $\{\mathbf{y}_{k+1}^{(i)}\}_{k\in\mathcal{K}}$ denote the corresponding measured model-variable vectors, with $\mathbf{x}_{k+1}^{(i)}$ denoting their constrained-state subvectors. Starting from $\widehat{\mathcal{H}}_0^{(i)}=\mathcal{H}_0^{(i)}$, the learned dynamics model is recursively rolled out under the recorded controls:
\begin{subequations}
\label{eq:recursive_rollout_model}
\begin{align}
\hat{\mathbf{y}}_{k+1}^{(i)}
&=
f_{\theta_x}\!\left(
\widehat{\mathcal{H}}_k^{(i)},
\mathbf{u}_{k,\mathrm{d}}^{(i)}
\right),
\qquad k\in\mathcal{K},\\
\widehat{\mathcal{H}}_{k+1}^{(i)}
&=
\operatorname{shift}\!\left(
\widehat{\mathcal{H}}_k^{(i)},
\hat{\mathbf{y}}_{k+1}^{(i)},
\mathbf{u}_{k,\mathrm{d}}^{(i)}
\right).
\end{align}
\end{subequations}
At each prediction stage, $\hat{\mathbf{x}}_{k+1}^{(i)}$ is obtained as the constrained-state subvector of $\hat{\mathbf{y}}_{k+1}^{(i)}$. Consequently, the resulting state-prediction errors include both one-step model mismatch and its propagation through the recursive model-variable rollout.

For each prediction stage $k=1,\ldots,N_p$, define the nonconformity score
\begin{equation}
s_{i,k}:=\left\|\mathbf{x}_k^{(i)}-\hat{\mathbf{x}}_k^{(i)}\right\|_\infty=\max_{j=1,\ldots,n_x}\left|x_{k,j}^{(i)}-\hat{x}_{k,j}^{(i)}\right|.
\label{eq:nonconformity_score}
\end{equation}
Thus, one scalar score records the largest absolute prediction error among the constrained state components at stage $k$.

\subsection{Stage-Wise Conformal Tightening}
\label{sec:conformal_tightening}

Given $m$ recursive-rollout scores at prediction stage $k$, let $s_{(1),k}\leq\cdots\leq s_{(m),k}$ denote their order statistics. For a prescribed miscoverage level $\alpha\in(0,1)$, define the conformal rank
\begin{equation}
q_\alpha=\left\lceil(m+1)(1-\alpha)\right\rceil.
\label{eq:cp_rank}
\end{equation}
The corresponding stage-wise margin is
\begin{equation}
r_k=\begin{cases}s_{(q_\alpha),k},&q_\alpha\leq m,\\+\infty,&q_\alpha=m+1,\end{cases}
\label{eq:cp_margin}
\end{equation}
for $k=1,\ldots,N_p$. The complete tightening sequence is
\begin{equation}
\mathbf{r}:=[r_1,\ldots,r_{N_p}]^\top.
\label{eq:tightening_sequence}
\end{equation}
For each stage, the corresponding tightened state set is
\begin{equation}
\mathcal{X}_k^{\mathrm{tight}}:=\left\{\mathbf{x}\in\mathbb{R}^{n_x}\mid\mathbf{x}_\ell+r_k\mathbf{1}_{n_x}\leq\mathbf{x}\leq\mathbf{x}_h-r_k\mathbf{1}_{n_x}\right\}.
\label{eq:tightened_state_set}
\end{equation}

Geometrically, $r_k$ is the radius of the calibrated $\ell_\infty$ error ball at prediction stage $k$. Because the admissible state set is box-constrained, the same scalar is also the margin by which each state bound is shifted inward in~\eqref{eq:tightened_state_set}. Accordingly, the terms \emph{calibrated error radius} and \emph{constraint-tightening margin} refer to the same quantity $r_k$, with the former emphasizing prediction uncertainty and the latter emphasizing its effect on the admissible state set.

The input set remains $\mathcal{U}$ because the tightening is introduced to account for state-prediction mismatch rather than uncertainty in the generated control values.

After the final dynamics model is selected and frozen, let $\mathcal{D}_{\mathrm{cal}}$ denote an independent calibration set. Applying~\eqref{eq:cp_rank}--\eqref{eq:cp_margin} to the recursive-rollout scores from $\mathcal{D}_{\mathrm{cal}}$ yields the final sequence $\mathbf{r}^*:=[r_1^*,\ldots,r_{N_p}^*]^\top$.

\begin{assumption}[Calibration--deployment exchangeability]
\label{ass:cp_exchangeability}
For the frozen final dynamics model $f_{\theta_x^*}$, the calibration scores and the corresponding score of a future deployment scenario are exchangeable at each prediction stage $k=1,\ldots,N_p$.
\end{assumption}

Under Assumption~\ref{ass:cp_exchangeability}, standard split conformal calibration gives the stage-wise marginal coverage~\cite{lei2018distribution,angelopoulos2023conformal}
\begin{equation}
\mathbb{P}\!\left(\left\|\mathbf{x}_k-\hat{\mathbf{x}}_k\right\|_\infty\leq r_k^*\right)\geq1-\alpha,\qquad k=1,\ldots,N_p.
\label{eq:stagewise_cp_coverage}
\end{equation}

At deployment, the controller generates an $N_p$-step control sequence but only its first action is applied before the system state is measured again and a new sequence is generated. Consequently, the plant-level constraint-satisfaction statement for the transition actually executed at the current control instant requires only the first-stage conformal event together with first-stage nominal constraint satisfaction. The later-stage coverage events are therefore not required in the probability bound developed below.

The complete stage-wise tightening sequence is nevertheless retained for policy synthesis. The policy is trained through an $N_p$-step virtual rollout, so tightened constraints at later prediction stages shape the complete predicted trajectory and control sequence and can consequently affect the first control action. In addition, recursive prediction errors generally vary with the look-ahead stage, making stage-dependent margins more informative than a single horizon-wide margin. This differs from~\cite{chee2024uncertainty}, where the estimated uncertainty is explicitly not propagated along the prediction horizon.

More broadly, in constraint-tightening and tube-based robust MPC, recursive-feasibility and stability analyses are typically constructed around the feasibility and evolution of finite-horizon predicted trajectories under tightened constraints or tubes, together with suitable shift, invariance, terminal, or contraction arguments. Retaining a horizon-indexed tightening sequence therefore provides a natural structure for future extensions of the present framework toward such analyses~\cite{mayne2005robust,mayne2011tube,kohler2018novel}.

\subsection{Tightening-Conditioned E2E-DPC Training}

During joint training, the same recursive score construction is applied to a separate validation set using the current dynamics model. At epoch $e$, these scores are used with~\eqref{eq:cp_rank}--\eqref{eq:cp_margin} to obtain a provisional stage-wise sequence $\mathbf{r}^{(e)}$. This sequence provides the tightening information used to condition the policy during that training epoch; the final deployment sequence $\mathbf{r}^*$ is computed separately from $\mathcal{D}_{\mathrm{cal}}$ after model selection.

At epoch $e$, the pOCP parameter vector in Section~\ref{sec:e2e_dpc_framework} is augmented by the provisional tightening sequence:
\begin{equation}
\boldsymbol{\eta}^{\mathrm{G},(e)}:=\left[\boldsymbol{\eta}^\top,\left(\mathbf{r}^{(e)}\right)^\top\right]^\top.
\label{eq:augmented_eta_guarantee}
\end{equation}
The policy is evaluated as
\begin{equation}
\mathbf{U}=f_{\theta_u}\!\left(\mathbf{y}_0,\boldsymbol{\eta}^{\mathrm{G},(e)}\right).
\label{eq:tightening_conditioned_policy}
\end{equation}
The virtual rollout and all optimization terms remain as defined in Section~\ref{sec:e2e_dpc_framework}, except that the state-constraint penalties use
\begin{equation}
\mathbf{x}_\ell+r_{k+1}^{(e)}\mathbf{1}_{n_x}\leq\hat{\mathbf{x}}_{k+1}\leq\mathbf{x}_h-r_{k+1}^{(e)}\mathbf{1}_{n_x},\qquad k\in\mathcal{K}.
\label{eq:training_tightened_constraints}
\end{equation}
The two-stage schedule, PCGrad update, validation criteria, and coupling coefficient $\beta$ are otherwise unchanged.

\subsection{Hoeffding Certification of Tightened Nominal Satisfaction}

After $f_{\theta_x^*}$, $f_{\theta_u^*}$, and $\mathbf{r}^*$ are fixed, let $\mathcal{D}_{\mathrm{cert}}$ denote a separate certification sample containing $m_{\mathrm{cert}}$ deployment scenarios. Each scenario is evaluated only at the first receding-horizon transition. Define the overall indicator
\begin{equation}
I_i:=\begin{cases}1,&\hat{\mathbf{x}}_1^{(i)}\in\mathcal{X}_1^{\mathrm{tight},*}\text{ and }\mathbf{u}_0^{(i)}\in\mathcal{U},\\0,&\text{otherwise},\end{cases}
\label{eq:indicator_tight}
\end{equation}
where $\mathbf{u}_0^{(i)}$ is the first element of the control sequence generated by $f_{\theta_u^*}$ and $\hat{\mathbf{x}}_1^{(i)}$ is the corresponding one-step nominal prediction. Let
\begin{equation}
\widehat{\mu}_{\mathrm{tight}}:=\frac{1}{m_{\mathrm{cert}}}\sum_{i=1}^{m_{\mathrm{cert}}}I_i,
\label{eq:empirical_tight_probability}
\end{equation}
where $\mu_{\mathrm{tight}}$ denotes the true deployment-scenario probability that the first nominal prediction and the first generated control satisfy the corresponding tightened state and input constraints.

\begin{assumption}[Certification-sample independence]
\label{ass:hoeffding_iid}
The indicators $I_1,\ldots,I_{m_{\mathrm{cert}}}$ are independent and identically distributed (i.i.d.) Bernoulli random variables.
\end{assumption}
\begin{theorem}[Hoeffding lower bound]
\label{thm:hoeffding_tight}
Under Assumption~\ref{ass:hoeffding_iid}, with confidence at least $1-\delta_{\mathrm{H}}$ over the certification sample,
\begin{equation}
\mu_{\mathrm{tight}}\geq\underline{\mu}_{\mathrm{tight}}:=\max\left\{0,\widehat{\mu}_{\mathrm{tight}}-\sqrt{\frac{\ln(1/\delta_{\mathrm{H}})}{2m_{\mathrm{cert}}}}\right\}.
\label{eq:hoeffding_tight_lb}
\end{equation}
\end{theorem}
\begin{pf}
Under Assumption~\ref{ass:hoeffding_iid}, $I_1,\ldots,I_{m_{\mathrm{cert}}}$ are i.i.d. Bernoulli random variables with mean $\mu_{\mathrm{tight}}$. Hoeffding's inequality~\cite{hoeffding1963probability} gives
\[
\mathbb{P}\!\left(\mu_{\mathrm{tight}}<\widehat{\mu}_{\mathrm{tight}}-a\right)\leq\exp(-2m_{\mathrm{cert}}a^2).
\]
Setting $a=\sqrt{\ln(1/\delta_{\mathrm{H}})/(2m_{\mathrm{cert}})}$ and clipping the lower bound at zero yields~\eqref{eq:hoeffding_tight_lb}. \qed
\end{pf}

\subsection{Combined Plant-Level Guarantee}

Define the events
\begin{equation}
\mathcal{A}:=\left\{\hat{\mathbf{x}}_1\in\mathcal{X}_1^{\mathrm{tight},*},\ \mathbf{u}_0\in\mathcal{U}\right\},
\label{eq:event_nominal_tight}
\end{equation}
and
\begin{equation}
\mathcal{B}:=\left\{\left\|\mathbf{x}_1-\hat{\mathbf{x}}_1\right\|_\infty\leq r_1^*\right\}.
\label{eq:event_cp_error}
\end{equation}
\begin{proposition}[Tightening implication]
\label{prop:tightening_implication}
If $\mathcal{A}\cap\mathcal{B}$ occurs, then the plant state after the applied control action satisfies $\mathbf{x}_1\in\mathcal{X}$ and the applied input satisfies $\mathbf{u}_0\in\mathcal{U}$.
\end{proposition}
\begin{pf}
On $\mathcal{B}$, the definition of the infinity norm implies
\[
\hat{\mathbf{x}}_1-r_1^*\mathbf{1}_{n_x}\leq\mathbf{x}_1\leq\hat{\mathbf{x}}_1+r_1^*\mathbf{1}_{n_x}.
\]
On $\mathcal{A}$,
\[
\mathbf{x}_\ell+r_1^*\mathbf{1}_{n_x}\leq\hat{\mathbf{x}}_1\leq\mathbf{x}_h-r_1^*\mathbf{1}_{n_x}.
\]
Combining these inequalities gives $\mathbf{x}_\ell\leq\mathbf{x}_1\leq\mathbf{x}_h$, while $\mathbf{u}_0\in\mathcal{U}$ follows directly from $\mathcal{A}$. \qed
\end{pf}
Let $\mu_{\mathrm{plant}}:=\mathbb{P}(\mathbf{x}_1\in\mathcal{X},\mathbf{u}_0\in\mathcal{U})$ denote the true deployment-scenario probability of original-constraint satisfaction for the applied receding-horizon transition.
\begin{theorem}[Combined plant-level guarantee]
\label{thm:combined_guarantee}
Suppose Assumptions~\ref{ass:cp_exchangeability} and~\ref{ass:hoeffding_iid} hold. Then, with confidence at least $1-\delta_{\mathrm{H}}$ over the certification sample,
\begin{equation}
\mu_{\mathrm{plant}}\geq\underline{\mu}_{\mathrm{plant}}:=\max\left\{0,\underline{\mu}_{\mathrm{tight}}-\alpha\right\}.
\label{eq:combined_plant_lb}
\end{equation}
\end{theorem}
\begin{pf}
By Proposition~\ref{prop:tightening_implication}, plant constraint satisfaction contains $\mathcal{A}\cap\mathcal{B}$. Hence
\[
\mu_{\mathrm{plant}}\geq\mathbb{P}(\mathcal{A}\cap\mathcal{B})\geq\mathbb{P}(\mathcal{A})-\mathbb{P}(\mathcal{B}^c).
\]
By definition, $\mathbb{P}(\mathcal{A})=\mu_{\mathrm{tight}}$, while the $k=1$ instance of~\eqref{eq:stagewise_cp_coverage} gives $\mathbb{P}(\mathcal{B}^c)\leq\alpha$. Therefore,
\[
\mu_{\mathrm{plant}}\geq\mu_{\mathrm{tight}}-\alpha.
\]
Theorem~\ref{thm:hoeffding_tight} gives $\mu_{\mathrm{tight}}\geq\underline{\mu}_{\mathrm{tight}}$ with confidence at least $1-\delta_{\mathrm{H}}$. Substitution and clipping at zero yield~\eqref{eq:combined_plant_lb}. \qed
\end{pf}

Algorithm~\ref{alg:guarantee_training} summarizes the complete procedure while reusing the training operations already defined in Algorithm~\ref{alg:e2e_training}.

\refstepcounter{algorithm}\label{alg:guarantee_training}
\begin{tcolorbox}[breakable,enhanced,title={Algorithm~\thealgorithm: E2E-DPC Training with Stage-Wise Conformal Tightening and Certification}]
{\small
\begin{algorithmic}[1]
\Require $\mathcal{D}_{\mathrm{tr}}$, $\mathcal{D}_{\mathrm{val}}$, $\mathcal{D}_{\mathrm{cal}}$, $\mathcal{D}_{\mathrm{cert}}$, $N_p$, $\alpha$, $\delta_{\mathrm{H}}$, and the initialization and hyperparameters of Algorithm~\ref{alg:e2e_training}
\State Execute Stage 1 of Algorithm~\ref{alg:e2e_training}
\For{each joint-training epoch $e$}
\State Compute recursive rollout errors on $\mathcal{D}_{\mathrm{val}}$ using the current $f_{\theta_x}$ and the recorded controls according to~\eqref{eq:recursive_rollout_model}--\eqref{eq:nonconformity_score}
\State Compute the provisional stage-wise sequence $\mathbf{r}^{(e)}$ using~\eqref{eq:cp_rank}--\eqref{eq:cp_margin}
\State Form $\boldsymbol{\eta}^{\mathrm{G},(e)}$ according to~\eqref{eq:augmented_eta_guarantee}
\State Execute the Stage 2 mini-batch updates of Algorithm~\ref{alg:e2e_training}, replacing the policy input by~\eqref{eq:tightening_conditioned_policy} and the original state bounds in $\mathcal{L}_{\mathrm{cons}}$ by~\eqref{eq:training_tightened_constraints}
\State Evaluate the validation losses and update $\beta$ using the same criteria as Algorithm~\ref{alg:e2e_training}
\EndFor
\State Select the final checkpoint and freeze $f_{\theta_x^*}$ and $f_{\theta_u^*}$
\State Compute the final fixed stage-wise sequence $\mathbf{r}^*$ on $\mathcal{D}_{\mathrm{cal}}$ using the frozen $f_{\theta_x^*}$ and~\eqref{eq:cp_rank}--\eqref{eq:cp_margin}
\For{each sample $i$ in $\mathcal{D}_{\mathrm{cert}}$}
\State Generate the control sequence using $f_{\theta_u^*}$ with $\mathbf{r}^*$, extract $\mathbf{u}_0^{(i)}$, compute the corresponding one-step nominal prediction $\hat{\mathbf{x}}_1^{(i)}$, and evaluate $I_i$ using~\eqref{eq:indicator_tight}
\EndFor
\State Compute $\widehat{\mu}_{\mathrm{tight}}$, $\underline{\mu}_{\mathrm{tight}}$, and $\underline{\mu}_{\mathrm{plant}}$ using~\eqref{eq:empirical_tight_probability}, \eqref{eq:hoeffding_tight_lb}, and~\eqref{eq:combined_plant_lb}
\Ensure Frozen $f_{\theta_x^*}$, frozen $f_{\theta_u^*}$, final tightening sequence $\mathbf{r}^*$, and the reported probabilistic certificate
\end{algorithmic}
}
\end{tcolorbox}

\section{Case Study}
\label{sec:case_study}

This section evaluates the proposed E2E-DPC framework and its probabilistic constraint-tightening mechanism in a high-fidelity building DR case study. The general pOCP in~\eqref{eq:pocp} is instantiated for a multi-zone residential building, allowing the E2E training procedure in Algorithm~\ref{alg:e2e_training} and the stage-wise conformal tightening and certification procedure in Algorithm~\ref{alg:guarantee_training} to be evaluated under realistic plant--model mismatch. Section~\ref{subsec:simulation_framework} introduces the building testbed and closed-loop simulation framework, Section~\ref{subsec:case_problem} specifies the corresponding control problem, and Section~\ref{subsec:constraint_tightening_analysis} evaluates the closed-loop effectiveness, certification behaviour, and tightening-feasibility trade-offs of the proposed framework.

\subsection{Building Testbed and Co-Simulation Framework}
\label{subsec:simulation_framework}

To ensure a realistic and reproducible evaluation, the co-simulation framework established in our previous study~\cite{xu2026input} is adapted to the DPC setting considered here. The building dynamics are simulated using EnergyPlus~\cite{crawley2001energyplus}, a widely adopted whole-building energy simulation program. The interaction between EnergyPlus and the Python-based experimental pipeline is implemented using Energym~\cite{scharnhorst2021energym}. Energym is an open-source building simulation library developed specifically for systematic and reproducible controller benchmarking, providing a common interface through which a Python program can transmit control actions to a building model and receive the resulting building states, energy measurements, and weather information.

The \texttt{ApartmentsThermal-v0} environment represents a four-story residential building located in Tarragona, Spain, with a total floor area of approximately $417.12~\mathrm{m}^2$ and an air volume of $1042.83~\mathrm{m}^3$~\cite{scharnhorst2021energym}. Each story corresponds to an individual apartment and contains two thermal zones, resulting in eight thermal zones in total. A single thermostat setpoint is assigned to each apartment, giving four independently controllable thermostat setpoints. Space heating and domestic hot-water production are supplied by a centralized water-to-water geothermal Heat Pump connected to a vertical ground heat exchanger. Indoor conditioning is provided through fan-coil units, with two units installed in each apartment.

Figure~\ref{fig:simulation_framework} illustrates the closed-loop co-simulation architecture. At each control instant, the EnergyPlus model returns the current building observations to the Python environment through the Energym interface. The DPC policy network $f_{\theta_u}$ receives the current building information together with the future electricity-price sequence and generates the complete future thermostat-setpoint sequence. For the constraint-tightened DPC variants, the stage-wise tightening sequence $\mathbf{r}^*$ developed in Section~\ref{sec:constraint_tightening} is additionally supplied to the policy. Consistent with the receding-horizon implementation described in Sections~\ref{sec:e2e_dpc_framework} and~\ref{sec:constraint_tightening}, only the first element of the generated control sequence is applied to the building. EnergyPlus then advances the building simulation and returns the updated observations, thereby completing the closed-loop control cycle. The learned dynamics model $f_{\theta_x}$ is required during offline model--policy training and conformal calibration, but is not part of the online control-action generation loop. Its predictions can therefore be retained for analysis and monitoring without affecting the applied control action.

\begin{figure*}[t]
    \centering
    \includegraphics[width=\textwidth]{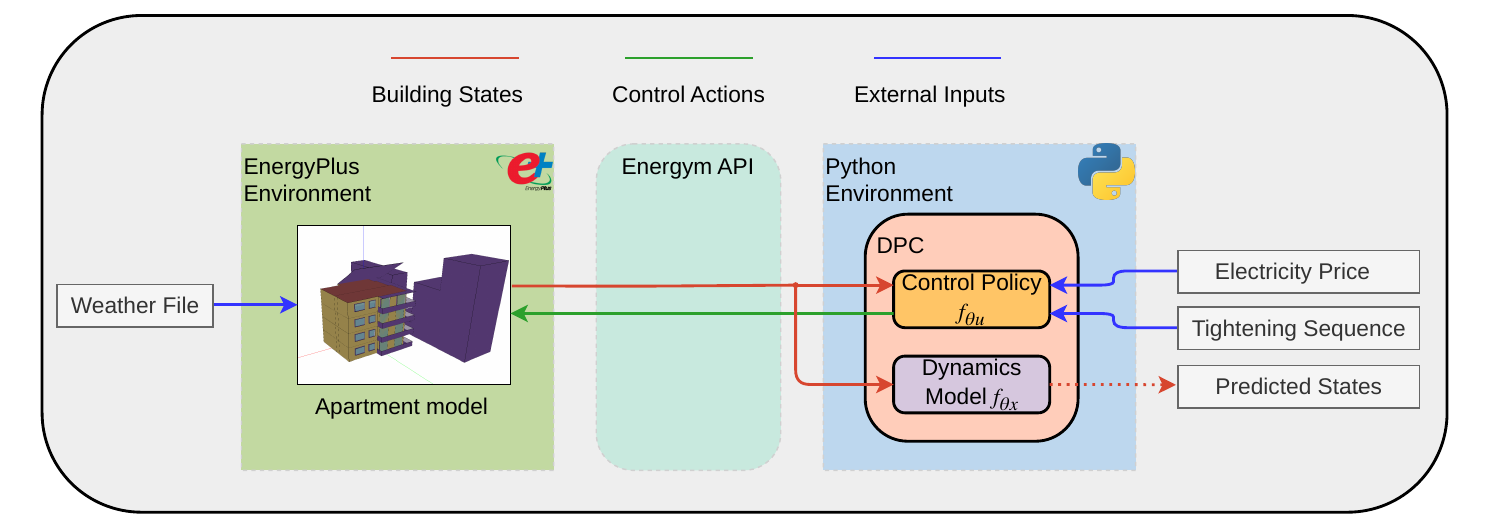}
    \caption{The co-simulation framework for online deployment of the DPC controller. The explicit control policy $f_{\theta_u}$ receives Building States (red line) from the EnergyPlus simulation and External Inputs (blue lines), including the future electricity-price sequence and, for the constraint-tightened controllers, the stage-wise tightening sequence $\mathbf{r}^*$. It generates the Control Actions (green line), which are returned to EnergyPlus through the Energym API. The dynamics model $f_{\theta_x}$ is used during offline training and conformal calibration but is not part of the online control-action generation loop; its predictions (dotted line) are retained for analysis and monitoring.}
    \label{fig:simulation_framework}
\end{figure*}

The building is operated under the Spanish residential time-of-use (TOU) tariff from Iberdrola Espa\~na's Plan Ahorro Solar~\cite{iberdrola_ahorro_solar}.
\begin{itemize}
    \item \textbf{Off-Peak / Valle (00:00--08:00):} \texteuro{}0.089/kWh on weekdays, with the off-peak rate applied throughout weekends.
    \item \textbf{Mid-Peak / Llano (08:00--10:00, 14:00--18:00, and 22:00--24:00):} \texteuro{}0.135/kWh on weekdays.
    \item \textbf{Peak / Punta (10:00--14:00 and 18:00--22:00):} \texteuro{}0.210/kWh on weekdays.
\end{itemize}
This structure incentivizes shifting flexible electricity use from the peak periods towards lower-price periods.

Among the weather files available in Energym, the \texttt{ESP\_CT\_Girona2} weather file is selected for the simulation experiments~\cite{scharnhorst2021energym}. All closed-loop evaluations are conducted from 6 to 12 January. The metadata of the corresponding EnergyPlus weather file identifies this interval as the ``Winter -- Week Nearest Min Temperature For Period'' extreme period, and this seven-day interval is therefore used as the winter evaluation period.

\subsection{Problem Formulation}
\label{subsec:case_problem}

The building control problem is formulated as a specific instance of the general pOCP in~\eqref{eq:pocp}. The objective is to minimize electricity expenditure over a finite prediction horizon of $N_p$ control steps while maintaining the zonal temperatures and thermostat setpoints within their prescribed operating constraints. Following Section~\ref{sec:pocp_formulation}, $\mathcal{K}=\{0,\ldots,N_p-1\}$ denotes the prediction-stage index set.

\subsubsection{Predictive Model and Policy Setup}

The predictive dynamics model $f_{\theta_x}$ adopts the EoT architecture introduced in Section~\ref{subsec:eot_dynamics}. It receives a fixed-length historical sequence of length $N_h$ containing the eight zonal temperatures (\texttt{Z01\_T}--\texttt{Z08\_T}), total facility electricity consumption (\texttt{Fa\_E\_All}), heat-pump return-water temperature (\texttt{Bd\_T\_HP\_return}), appliance electricity consumption (\texttt{Fa\_E\_Appl}), and the four thermostat setpoints. The model jointly predicts the complete 11-dimensional model-variable vector $\hat{\mathbf{y}}_k$, consisting of the eight zonal temperatures together with \texttt{Fa\_E\_All}, \texttt{Bd\_T\_HP\_return}, and \texttt{Fa\_E\_Appl}. The eight zonal temperatures form the constrained-state subvector $\hat{\mathbf{x}}_k$ of $\hat{\mathbf{y}}_k$, while the remaining three components are auxiliary outputs recursively propagated by the EoT model. Among them, the predicted total facility electricity consumption is additionally used to evaluate the economic objective.

The explicit control policy $f_{\theta_u}$ is implemented as multilayer perceptron (MLP). Consistent with the formulation in Section~3.2.2, the policy is evaluated once at the beginning of each prediction horizon and generates the complete future control sequence
\begin{equation}
\mathbf{U}=[\mathbf{u}_0^\top,\ldots,\mathbf{u}_{N_p-1}^\top]^\top \in \mathbb{R}^{4N_p}.
\end{equation}
For the DPC variants without constraint tightening, the policy input contains the current measured model-variable vector $\mathbf{y}_0\in\mathbb{R}^{11}$, corresponding to the prediction targets of $f_{\theta_x}$, together with the $N_p$-step future electricity-price sequence. At deployment, the constraint-tightened variants additionally receive the final stage-wise conformal tightening sequence $\mathbf{r}^*=[r_1^*,\ldots,r_{N_p}^*]^\top$, obtained from $\mathcal{D}_{\mathrm{cal}}$ after model selection as described in Section~\ref{sec:constraint_tightening}. Although the policy generates all $N_p$ future control actions simultaneously, only $\mathbf{u}_0$ is applied before the building state is measured again and a new control sequence is generated.

\subsubsection{Control Problem Specification}

The controller operates over a finite prediction horizon of $N_p$ control steps and at a 15-minute control interval. The variables in the general formulation of Section~\ref{sec:e2e_dpc_framework} are instantiated as follows:
\begin{itemize}
    \item The constrained state vector $\hat{\mathbf{x}}_k\in\mathbb{R}^{8}$ is the zonal-temperature subvector of $\hat{\mathbf{y}}_k$ and comprises the eight predicted zonal temperatures,
    \begin{equation}
    \hat{\mathbf{x}}_k=
    [\widehat{\texttt{Z01\_T}}_k,\ldots,\widehat{\texttt{Z08\_T}}_k]^\top.
    \end{equation}

    \item The control input vector $\mathbf{u}_k\in\mathbb{R}^{4}$ contains the four apartment thermostat setpoints,
    \begin{equation}
    \mathbf{u}_k=
    [\texttt{P1\_T\_Thermostat\_sp},\ldots,\texttt{P4\_T\_Thermostat\_sp}]^\top.
    \end{equation}

    \item In addition to the constrained zonal temperatures, $f_{\theta_x}$ predicts \texttt{Fa\_E\_All}, \texttt{Fa\_E\_Appl}, and \texttt{Bd\_T\_HP\_return} as auxiliary components of $\hat{\mathbf{y}}_k$ for recursive rollout. These variables are not subject to the thermal-comfort constraints. In particular, $\hat{E}_{k+1}$, corresponding to the \texttt{Fa\_E\_All} component of $\hat{\mathbf{y}}_{k+1}$, defines the economic stage cost associated with control stage $k$ as
    \begin{equation}
    \ell(\hat{\mathbf{y}}_{k+1},\mathbf{u}_k,\pi_k)=\pi_k\,\operatorname{ReLU}(\hat{E}_{k+1}),\qquad k\in\mathcal{K},
    \label{eq:case_stage_cost}
    \end{equation}
    where $\pi_k$ is the TOU electricity price associated with the corresponding control interval. The ReLU operator prevents the jointly trained dynamics model from reducing $\mathcal{L}_{\mathrm{obj}}$ through physically inadmissible negative energy predictions.

    \item The admissible state and input sets are defined by the thermal-comfort and thermostat operating limits
    \begin{equation}
    19\,^{\circ}\mathrm{C}\leq \hat{x}_{k,j}\leq24\,^{\circ}\mathrm{C},\qquad j=1,\ldots,8,\quad k=1,\ldots,N_p.
    \end{equation}
    and
    \begin{equation}
    16\,^{\circ}\mathrm{C}\leq u_{k,j}\leq26\,^{\circ}\mathrm{C},\qquad j=1,\ldots,4,\quad k\in\mathcal{K}.
    \end{equation}
    respectively.

    \item In this case study, the time-varying horizon-level parameter supplied to the policy is the future TOU electricity-price sequence. The state and input bounds remain fixed throughout the experiments and are therefore not explicitly concatenated to the policy input. For the constraint-tightened formulation, the policy is additionally conditioned on $\mathbf{r}^*$, and the nominal zonal-temperature constraints at prediction stage $k$ are tightened according to
    \begin{equation}
    \mathbf{x}_{\ell}+r_k^*\mathbf{1}_{8}
    \leq
    \hat{\mathbf{x}}_k
    \leq
    \mathbf{x}_{h}-r_k^*\mathbf{1}_{8},
    \qquad k=1,\ldots,N_p,
    \end{equation}
    while the original input set $\mathcal{U}$ remains unchanged, as established in Section~\ref{sec:conformal_tightening}.
    
\end{itemize}

\subsection{Constraint Tightening and Certification Analysis}
\label{subsec:constraint_tightening_analysis}

This subsection evaluates the proposed framework through three complementary analyses. First, a controlled stress test examines the closed-loop effect of conformal tightening across MPC and DPC controllers. Second, E2E-DPC and Conventional DPC are compared across multiple loss-weight configurations to assess their stage-wise calibrated error radii and plant-level constraint-satisfaction certificates. Finally, the effect of prediction-horizon length on tightening feasibility is examined to reveal the associated trade-off and the role of multi-step identification.

All experiments use a historical window length of $N_h=10$. The EoT dynamics model uses an embedding dimension of 64, one attention head, a feed-forward dimension of 64, one encoder layer, and a dropout rate of 0.1. The DPC policy is implemented as an MLP with two hidden layers of 64 neurons. A batch size of 256 and a maximum of 2000 training epochs are used throughout. For conventional DPC, the dynamics model is first trained using the multi-step identification loss without an additional constraint penalty and is subsequently frozen during policy training, where the constraint-loss weight is fixed at $w_{\mathrm{cons}}=10$. For E2E-DPC, $w_{\mathrm{id}}=1$ and $w_{\mathrm{cons}}=10$ are fixed, while $w_{\mathrm{ms}}$ and $w_{\mathrm{obj}}$ are varied according to the experiment under consideration.

For all constraint-tightened experiments, split-conformal calibration uses a coverage level of $1-\alpha=0.99$, and the independent Hoeffding certificate is evaluated at a confidence level of $1-\delta_H=0.99$.

\subsubsection{Constraint-Tightening Stress Test Across MPC and DPC}
\label{subsubsec:stress_test}

For the constraint-tightening stress test, six controllers are organized into three matched pairs: Online MPC and Tightened Online MPC, Conventional DPC and Tightened Conventional DPC, and E2E-DPC and Tightened E2E-DPC. The Conventional DPC baseline follows the original two-stage formulation of Drgo\v{n}a et al.~\cite{drgovna2022differentiable}, in which the dynamics model is identified independently and subsequently frozen during policy optimization. E2E-DPC is trained according to Algorithm~\ref{alg:e2e_training}, while Tightened E2E-DPC follows Algorithm~\ref{alg:guarantee_training}. The stage-wise conformal tightening developed in Section~\ref{sec:conformal_tightening} is also applied to the Online MPC and Conventional DPC baselines. Both Online MPC variants retain non-negative slack variables on the predicted temperature constraints to preserve optimization feasibility, which is particularly relevant when tightening substantially reduces the admissible nominal state region.

All six controllers use $N_p=4$, corresponding to a one-hour prediction horizon. Within each pair, the tightened and untightened variants use the same model architecture, prediction horizon, and controller-specific loss settings, such that the introduction of conformal tightening is the principal methodological difference. For the DPC stress-test configurations, Conventional DPC uses $w_{\mathrm{cons}}=10$ and $w_{\mathrm{obj}}=1.5$, while E2E-DPC uses $w_{\mathrm{id}}=1$, $w_{\mathrm{ms}}=0.25$, $w_{\mathrm{cons}}=10$, and $w_{\mathrm{obj}}=1.5$. These settings are deliberately selected to form a constraint stress test, encouraging operation close to the admissible temperature boundaries so that plant--model mismatch produces observable violations in the untightened controllers.

Closed-loop performance is evaluated using the total electricity bill, thermal constraint violation, maximum temperature-bound violation, and online decision time. The reported thermal violation is evaluated against the original admissible state set $\mathcal{X}$ in~\eqref{eq:pocp_state_constraints}, rather than the tightened nominal sets defined in~\eqref{eq:tightened_state_set}, and is accumulated in degree-hours. Decision time is reported as mean $\pm$ standard deviation over all control instants. Table~\ref{tab:stress_test_results} summarizes these scalar metrics, while Figure~\ref{fig:stress_test_temperatures} shows the corresponding seven-day zonal-temperature trajectories.

\begin{table*}[t]
\centering
\caption{Closed-loop performance of the six controllers in the constraint-tightening stress test.}
\label{tab:stress_test_results}
\begin{tabular}{lrrrr}
\hline
Controller & Bill (\texteuro{}) & Violation ($^\circ$C$\cdot$h) & Max. violation ($^\circ$C) & Decision time (s) \\
\hline
Online MPC & 222.65 & 15.301 & 1.081 & $2.376 \pm 1.479$ \\
Tightened Online MPC & 229.21 & 0.090 & 0.254 & $12.091 \pm 14.313$ \\
Conventional DPC & 196.83 & 55.426 & 1.198 & $0.000450 \pm 0.002334$ \\
Tightened Conventional DPC & 246.80 & 0.000 & 0.000 & $0.000439 \pm 0.001900$ \\
E2E-DPC & 199.65 & 277.057 & 1.782 & $0.000427 \pm 0.001888$ \\
Tightened E2E-DPC & 248.66 & 0.000 & 0.000 & $0.000420 \pm 0.002238$ \\
\hline
\end{tabular}
\end{table*}

\begin{figure*}[t]
    \centering
    \includegraphics[width=\textwidth]{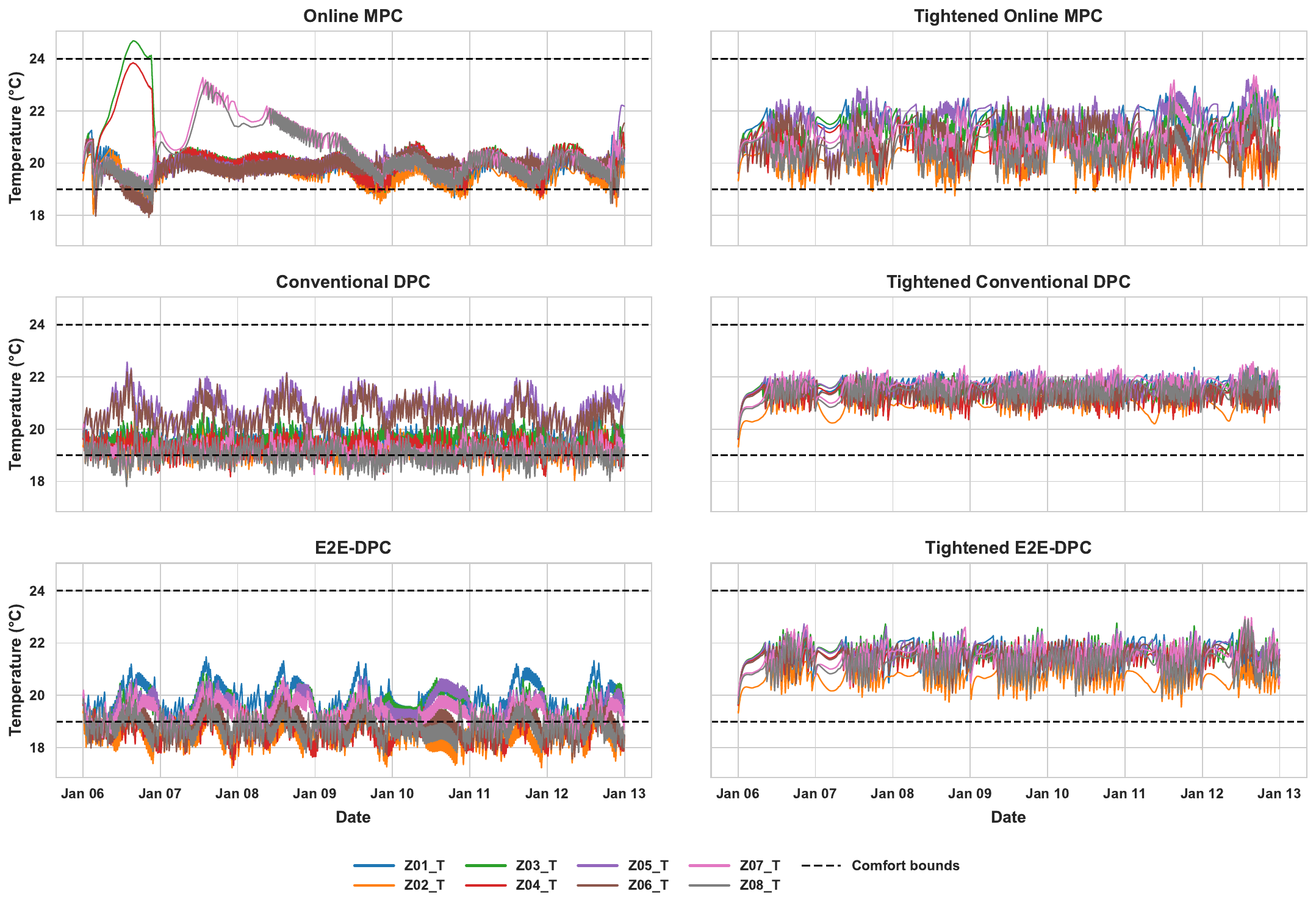}
    \caption{Closed-loop zonal-temperature trajectories for the six controllers in the constraint-tightening stress test with $N_p=4$. The dashed lines denote the $19$--$24\,^{\circ}\mathrm{C}$ thermal-comfort bounds. The left column shows the untightened controllers and the right column shows their constraint-tightened counterparts.}
    \label{fig:stress_test_temperatures}
\end{figure*}

Table~\ref{tab:stress_test_results} and Figure~\ref{fig:stress_test_temperatures} show a consistent effect of constraint tightening across the three controller families. All untightened controllers violate the original thermal constraints, accumulating $15.301$, $55.426$, and $277.057\,^{\circ}\mathrm{C}\!\cdot\!\mathrm{h}$ for Online MPC, Conventional DPC, and E2E-DPC, respectively. After tightening, both DPC variants achieve zero thermal violation over the seven-day evaluation, while Tightened Online MPC reduces the violation from $15.301$ to $0.090\,^{\circ}\mathrm{C}\!\cdot\!\mathrm{h}$, a reduction of approximately $99.4\%$. The small residual violation of Tightened Online MPC is compatible with the probabilistic nature of the conformal coverage and its soft-constrained MPC formulation, in which slack variables may relax the tightened nominal constraints to preserve optimization feasibility.

The improved constraint satisfaction is accompanied by higher electricity expenditure. Relative to the corresponding untightened controllers, the bill increases by approximately $2.9\%$, $25.4\%$, and $24.6\%$ for Online MPC, Conventional DPC, and E2E-DPC, respectively. The two DPC formulations nevertheless remain economically comparable, with their untightened bills differing by approximately $1.4\%$ and their tightened bills by less than $0.8\%$. Thus, the central result of this stress test is that conformal tightening substantially suppresses plant-level thermal violations across all three control paradigms, at the expected cost of restricting economically aggressive operation near the original constraint boundaries.

The online computation results further distinguish explicit DPC from optimization-based MPC. The four DPC variants require approximately $0.4$~ms per control decision, with essentially no additional online cost introduced by tightening. In contrast, Online MPC requires $2.38\pm1.48$~s per update, increasing to $12.09\pm14.31$~s for Tightened Online MPC. Although both remain well below the 15-minute control interval, the DPC policies avoid online finite-horizon optimization by transferring this computation to offline training. This separation may become increasingly relevant for larger online optimization problems; our previous MPC study observed substantially increasing solution times as the prediction horizon was extended~\cite{xu2026input}. The present $N_p=4$ experiment, however, is interpreted only as evidence of reduced online computation rather than as a direct scalability benchmark.

The calibrated stage-wise tightening sequences used in this stress test are
\begin{subequations}
\label{eq:stress_test_radii}
\begin{align}
\mathbf{r}^{*}_{\mathrm{MPC}}
&=
[1.123,\;1.631,\;2.040,\;2.351]^\top\,^{\circ}\mathrm{C},
\label{eq:stress_test_radii_mpc}\\
\mathbf{r}^{*}_{\mathrm{CDPC}}
&=
[1.665,\;1.840,\;1.952,\;2.073]^\top\,^{\circ}\mathrm{C},
\label{eq:stress_test_radii_cdpc}\\
\mathbf{r}^{*}_{\mathrm{E2E}}
&=
[1.173,\;1.725,\;2.091,\;2.351]^\top\,^{\circ}\mathrm{C}.
\label{eq:stress_test_radii_e2e}
\end{align}
\end{subequations}

In this stress-test configuration, E2E-DPC requires a substantially smaller first-stage margin than Conventional DPC, with $r_1^*$ decreasing from $1.665\,^{\circ}\mathrm{C}$ to $1.173\,^{\circ}\mathrm{C}$, corresponding to a reduction of approximately $29.5\%$. This observation is consistent with the loss design in~\eqref{eq:e2e_loss}--\eqref{eq:dpc_loss}: while the Conventional-DPC dynamics model is identified independently using the multi-step prediction objective before being frozen, E2E-DPC additionally employs the one-step identification term $\mathcal{L}_{\mathrm{id}}$ in~\eqref{eq:id_loss} and jointly refines the dynamics representation during policy optimization. The smaller $r_1^*$ enlarges the first-stage tightened set $\mathcal{X}_1^{\mathrm{tight},*}$ and is therefore favourable for the Hoeffding indicator in~\eqref{eq:indicator_tight}, which evaluates the first receding-horizon transition. Since the present experiment is intended as a targeted stress test of the tightening mechanism rather than a nominal performance ranking among the controller families, the closed-loop trajectories are emphasized here, while the stage-wise margins and the corresponding Hoeffding certificates are analyzed systematically over the extensive hyperparameter sweep in Section~\ref{subsubsec:e2e_certification}.

\subsubsection{Effect of End-to-End Training on Tightening and Certification}
\label{subsubsec:e2e_certification}

Building on the single-configuration observation in Section~\ref{subsubsec:stress_test}, this experiment examines whether the tightening and certification characteristics of E2E-DPC persist across a broader range of loss-weight settings. The prediction horizon is fixed at $N_p=4$ to focus the comparison on the training paradigm. For E2E-DPC, $w_{\mathrm{ms}}\in\{0.25,0.5,1,2,4\}$ is combined with the 16 objective-loss weights
\[
\begin{aligned}
\mathcal{W}_{\mathrm{obj}}=\{&
0,0.05,0.1,0.25,0.5,0.75,1,1.25,\\
&1.5,1.75,2,2.5,3,3.5,4,5\},
\end{aligned}
\]
giving 80 configurations. Conventional DPC is evaluated using the same 16 values in $\mathcal{W}_{\mathrm{obj}}$. At the fixed prediction horizon, its pretrained dynamics model is frozen before policy optimization and therefore yields one fixed stage-wise tightening sequence, whereas the jointly trained E2E dynamics model yields configuration-dependent sequences.

Figure~\ref{fig:stagewise_tightening_h4} compares the stage-wise calibrated error radii obtained under the two training paradigms. All 80 E2E-DPC configurations yield smaller $r_1^*$ and $r_2^*$ than Conventional DPC. The mean tightening margin at the first prediction stage decreases from $1.665\,^{\circ}\mathrm{C}$ to $1.149\,^{\circ}\mathrm{C}$, corresponding to a reduction of approximately $31.0\%$, while the mean second-stage reduction is approximately $9.4\%$. This advantage weakens later in the recursive rollout: only 29/80 and 7/80 E2E-DPC configurations have smaller third- and fourth-stage calibrated error radii, respectively.

This stage-dependent behaviour is consistent with the loss design in~\eqref{eq:e2e_loss}--\eqref{eq:dpc_loss}. The one-step identification term $\mathcal{L}_{\mathrm{id}}$ in~\eqref{eq:id_loss} anchors local predictive fidelity, while $\mathcal{L}_{\mathrm{ms}}$ in~\eqref{eq:ms_loss} penalizes prediction-error accumulation over recursive rollouts; the control-oriented terms further refine the shared dynamics representation during joint training. The results therefore show a consistent early-stage reduction in calibrated prediction error under E2E training, while this advantage does not necessarily persist throughout the prediction horizon.

\begin{figure}[t]
    \centering
    \includegraphics[width=0.95\columnwidth]{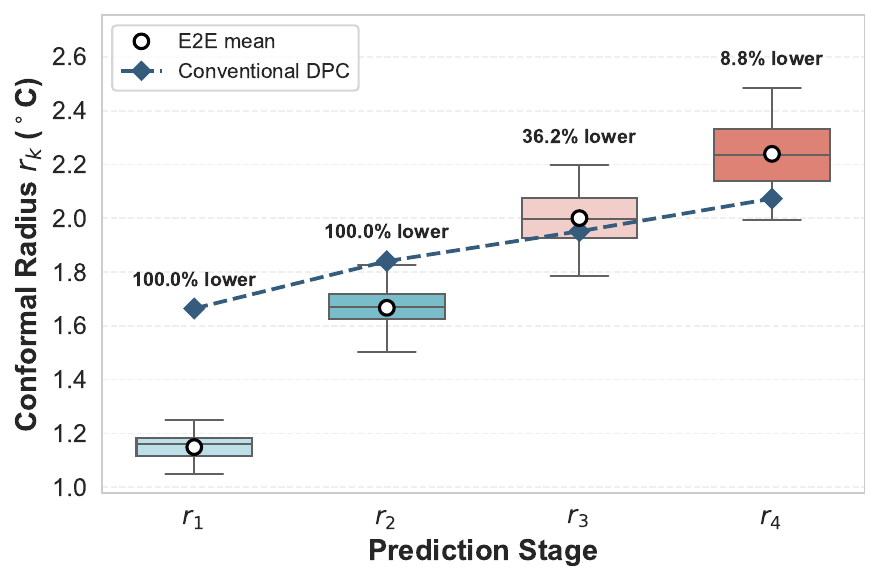}
    \caption{Stage-wise calibrated error radii at $N_p=4$, which are equivalently the constraint-tightening margins applied in~\eqref{eq:tightened_state_set}. Boxplots summarize the 80 E2E-DPC hyperparameter configurations, while the dashed curve denotes the fixed Conventional-DPC tightening sequence. White circles indicate the mean E2E-DPC radius at each stage, and the annotations report the percentage of E2E-DPC configurations with a smaller radius than Conventional DPC.}
    \label{fig:stagewise_tightening_h4}
\end{figure}

Before evaluating the certificates, a configuration is termed \emph{tightening-feasible} if the tightened set $\mathcal{X}^{\mathrm{tight},*}_k$ defined in~\eqref{eq:tightened_state_set} remains non-empty at every prediction stage.

The formal indicator $I_i$ in~\eqref{eq:indicator_tight} is deliberately stringent: a certification sample succeeds only when the first nominal temperatures of all eight zones satisfy the tightened state constraints and all four thermostat inputs satisfy $\mathcal{U}$. To facilitate analysis of this joint criterion, eight additional zone-wise diagnostics are evaluated, each requiring one zonal temperature and its associated apartment thermostat input to satisfy their corresponding first-stage constraints.

A separate Hoeffding lower bound is computed for each zone-wise diagnostic. To retain a simultaneous confidence level of $1-\delta_H=0.99$ across all eight bounds, a Bonferroni correction assigns a failure probability of $\delta_H/8$ to each bound, giving a per-zone confidence of $0.99875$. The resulting plant-level lower bounds $\underline{\mu}_{\mathrm{plant},z}$ are summarized by
\begin{equation}
\bar{\mu}_{\mathrm{zone}}^{\mathrm{LB}}
=
\frac{1}{8}
\sum_{z=1}^{8}
\underline{\mu}_{\mathrm{plant},z}.
\label{eq:mean_zone_certificate}
\end{equation}
This quantity is reported only as a diagnostic; the formal guarantee remains the joint plant-level certificate in~\eqref{eq:combined_plant_lb}.

Table~\ref{tab:certificate_comparison_h4} shows that Conventional DPC remains tightening-feasible for all 16 configurations, while 76 of the 80 E2E-DPC configurations retain non-empty tightened sets throughout the horizon. The difference in certification is substantially larger. Conventional DPC has a median overall plant-level lower bound of zero and a maximum of only $0.0027$, despite a median zone-average bound of $0.419$. In contrast, E2E-DPC achieves a median overall lower bound of $0.814$, a maximum of $0.864$, and a median zone-average bound of $0.918$; 47 of the 80 configurations achieve $\underline{\mu}_{\mathrm{plant}}\geq0.80$. The near-zero overall Conventional-DPC certificate reflects the joint event in~\eqref{eq:indicator_tight}, rather than any compounding of the $0.99$ conformal coverage or Hoeffding confidence across zones. The median is reported because the certificate distribution contains a pronounced low-certificate tail, while the number exceeding $0.80$ indicates that the E2E advantage is not confined to a small number of selected configurations.

\begin{table}[t]
\centering
\caption{Tightening-feasibility and certification results at $N_p=4$.}
\label{tab:certificate_comparison_h4}
\small
\setlength{\tabcolsep}{3pt}
\begin{tabular}{lcc}
\hline
Metric & Conventional DPC & E2E-DPC \\
\hline
Tightening-feasible & 16/16 & 76/80 \\
Median $\underline{\mu}_{\mathrm{plant}}$ & 0.000 & 0.814 \\
Max. $\underline{\mu}_{\mathrm{plant}}$ & 0.0027 & 0.864 \\
Median $\bar{\mu}_{\mathrm{zone}}^{\mathrm{LB}}$ & 0.419 & 0.918 \\
$\underline{\mu}_{\mathrm{plant}}\geq0.80$ & 0 & 47 \\
\hline
\end{tabular}
\end{table}

Taken together, the results link the E2E training design in Section~\ref{sec:e2e_dpc_framework} directly to the certification mechanism in Section~\ref{sec:constraint_tightening}. E2E-DPC consistently produces a smaller tightening margin at the first prediction stage, $r_1^*$, while the formal certificate in~\eqref{eq:indicator_tight}--\eqref{eq:combined_plant_lb} provides a lower bound on plant-level constraint satisfaction for the receding-horizon transition actually applied. A smaller $r_1^*$ enlarges $\mathcal{X}_1^{\mathrm{tight},*}$ and makes the certified nominal event less restrictive, which is consistent with the substantially higher certified lower bounds obtained by E2E-DPC. At later prediction stages, however, the advantage in calibrated error radius diminishes as recursive prediction errors accumulate. Section~\ref{subsubsec:horizon_tradeoff} therefore examines the resulting trade-off between stronger certification of plant-level constraint satisfaction and preservation of nonempty tightened state sets as the prediction horizon increases, together with the role of $\mathcal{L}_{\mathrm{ms}}$ in mitigating this trade-off.

\subsubsection{Trade-off Between Horizon-Dependent Uncertainty Calibration and Tightened-Set Feasibility}
\label{subsubsec:horizon_tradeoff}

To examine how the tightening behaviour observed in Section~\ref{subsubsec:e2e_certification} changes with the prediction horizon, $N_p\in\{4,6,8\}$ is considered, corresponding to look-ahead periods of 1, 1.5, and 2 hours. Using the tightening-feasibility criterion introduced in Section~\ref{subsubsec:e2e_certification}, Table~\ref{tab:horizon_feasibility} summarizes the results.

\begin{table}[t]
\centering
\caption{Tightened-set feasibility across prediction horizons.}
\label{tab:horizon_feasibility}
\begin{tabular}{ccc}
\hline
$N_p$ & Conventional DPC & E2E-DPC \\
\hline
4 & 16/16 (100\%) & 76/80 (95.0\%) \\
6 & 16/16 (100\%) & 35/80 (43.8\%) \\
8 & 16/16 (100\%) & 4/80 (5.0\%) \\
\hline
\end{tabular}
\end{table}

A clear difference emerges between the two training paradigms. Conventional DPC remains tightening-feasible for all tested configurations, whereas the feasible fraction of E2E-DPC decreases from $95.0\%$ at $N_p=4$ to $5.0\%$ at $N_p=8$. This behaviour is consistent with their different training objectives. Conventional DPC identifies $f_{\theta_x}$ independently of the downstream control objective, with strong emphasis on recursive predictive consistency before the model is frozen. E2E-DPC instead jointly balances $\mathcal{L}_{\mathrm{id}}$, $\mathcal{L}_{\mathrm{ms}}$, $\mathcal{L}_{\mathrm{cons}}$, and $\mathcal{L}_{\mathrm{obj}}$ through~\eqref{eq:e2e_loss}--\eqref{eq:dpc_loss}. As shown in Section~\ref{subsubsec:e2e_certification}, this joint formulation produces substantially smaller tightening margins at the early prediction stages and markedly higher certified lower bounds on plant-level constraint satisfaction. However, as the prediction horizon increases, the need to balance prediction-oriented and control-oriented objectives can allow larger recursive errors to emerge at later stages, causing the corresponding tightened state sets to become empty more frequently. Relative to Conventional DPC, the observed trade-off is therefore between stronger plant-level constraint-satisfaction certification and the ability to preserve nonempty tightened state sets as the prediction horizon increases.

The role of $\mathcal{L}_{\mathrm{ms}}$ in this trade-off is illustrated in Figure~\ref{fig:wms_feasibility}. At $N_p=6$, increasing $w_{\mathrm{ms}}$ from $0.25$ to $4$ increases the number of tightening-feasible E2E-DPC configurations from 0/16 to 14/16; at $N_p=8$, only $w_{\mathrm{ms}}=4$ retains any feasible configurations, with 4/16. Consistently, the mean terminal-stage calibrated error radius decreases from $2.546$ to $2.358\,^{\circ}\mathrm{C}$ at $N_p=6$ and from $2.698$ to $2.420\,^{\circ}\mathrm{C}$ at $N_p=8$. These results agree with the role of $\mathcal{L}_{\mathrm{ms}}$ in~\eqref{eq:ms_loss}, which explicitly penalizes prediction-error accumulation during recursive rollout and therefore helps preserve nonempty tightened state sets as the prediction horizon increases.

\begin{figure}[t]
    \centering
    \includegraphics[width=0.95\columnwidth]{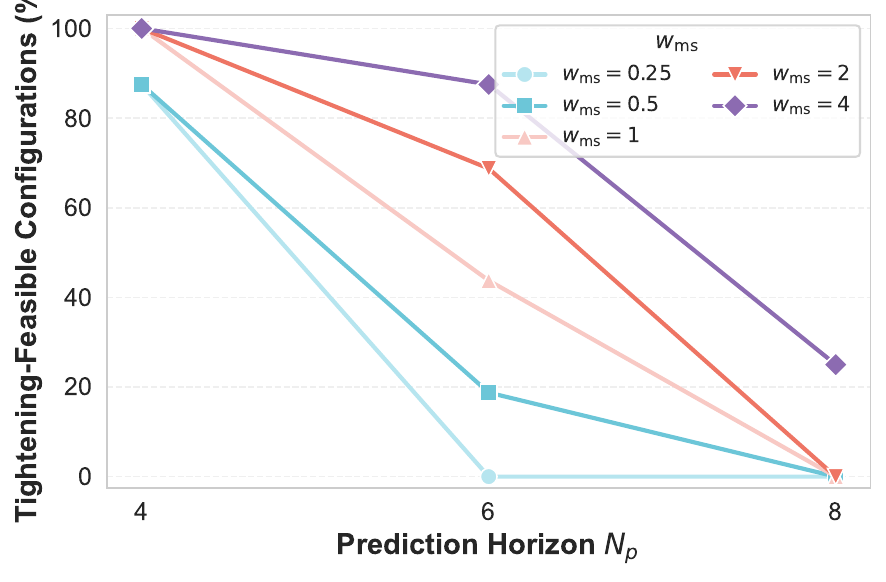}
    \caption{Effect of $w_{\mathrm{ms}}$ on E2E-DPC tightened-set feasibility as the prediction horizon increases. Each point gives the fraction of the 16 explored $w_{\mathrm{obj}}$ settings for which the tightened state set remains nonempty at every prediction stage for the corresponding $(N_p,w_{\mathrm{ms}})$ pair.}
    \label{fig:wms_feasibility}
\end{figure}

The results also reveal a second trade-off arising from the tightening design itself. In the proposed framework, each $r_k^*$ is calibrated from the recursive $k$-step prediction error through~\eqref{eq:nonconformity_score}--\eqref{eq:cp_margin}. The resulting horizon-indexed sequence therefore retains the observed growth of prediction uncertainty along the recursive rollout. Although only the first control action is applied, the later-stage tightened constraints shape the complete virtual trajectory and control sequence during policy synthesis, and can consequently influence the first action. This horizon-indexed structure is also naturally compatible with future recursive-feasibility and stability extensions based on trajectory shifting, invariance, terminal ingredients, or contraction arguments, as discussed in Section~\ref{sec:conformal_tightening}. Its cost is increased conservatism at later stages: as demonstrated above, sufficiently large $r_k^*$ can eventually render $\mathcal{X}_k^{\mathrm{tight},*}$ empty.

Chee et al.~\cite{chee2024uncertainty} make a different design choice. Their online weighted CP procedure computes a current uncertainty vector and forms a single dynamically tightened state set, without propagating the uncertainty magnitude with the look-ahead stage; the authors note that this choice helps promote feasibility in their finite-horizon formulation. In contrast, the proposed method calibrates the stage-wise sequence $\mathbf{r}^*=[r_1^*,\ldots,r_{N_p}^*]^\top$ directly from recursive prediction errors, thereby preserving the observed horizon dependence of plant--model mismatch. Although this can make later-stage tightened sets increasingly restrictive, the complete sequence shapes the virtual trajectory used for policy synthesis, allowing horizon-dependent uncertainty to influence the generated control sequence. Moreover, the resulting horizon-indexed constraint structure is naturally compatible with future recursive-feasibility and stability extensions based on trajectory shifting, invariance, terminal ingredients, or contraction arguments, as discussed in Section~\ref{sec:conformal_tightening}. The comparison therefore reflects a design trade-off: avoiding horizon-wise uncertainty propagation favours feasibility, whereas stage-wise calibration retains information about recursive uncertainty evolution and preserves a structured finite-horizon tightening formulation. Importantly, the formal plant-level certificate remains based only on the first executed transition through~\eqref{eq:indicator_tight}--\eqref{eq:combined_plant_lb}; the later-stage radii shape policy synthesis rather than being compounded into the certified probability bound.

\section{Conclusion and Future Work}
\label{sec:conclusion}

This paper proposed an E2E-DPC framework that jointly optimizes the learned dynamics model and explicit control policy while incorporating data-driven probabilistic constraint tightening. The E2E objective combines local identification, recursive multi-step prediction, constraint satisfaction, and economic performance, allowing the dynamics representation to be shaped by both prediction fidelity and the downstream control task. To account for plant--model mismatch, stage-wise calibrated error radii are obtained from recursive prediction errors and used as constraint-tightening margins throughout the finite-horizon policy rollout. Combined with independent Hoeffding certification, this yields a high-confidence plant-level guarantee on satisfaction of the original state and input constraints for the receding-horizon transition actually applied.

The case study demonstrates the complementary benefits of these two components. Constraint tightening substantially suppresses closed-loop thermal violations across Online MPC, Conventional DPC, and E2E-DPC, while the explicit DPC policies retain very low online decision times. Compared with Conventional DPC, E2E training consistently produces smaller tightening margins at the early prediction stages and substantially higher certified lower bounds on plant-level constraint satisfaction; at $N_p=4$, the median certified lower bound increases from 0 for Conventional DPC to 0.814 for E2E-DPC. At the same time, the experiments reveal an important trade-off: as the prediction horizon increases, larger late-stage calibrated errors can render the tightened state sets empty more frequently, although increasing the contribution of $\mathcal{L}_{\mathrm{ms}}$ substantially mitigates this effect by suppressing recursive prediction-error accumulation.

The probabilistic guarantee remains subject to the calibration--deployment exchangeability and certification-sample assumptions introduced in Section~\ref{sec:constraint_tightening}, and the current validation is limited to a high-fidelity simulation environment. Future work will therefore investigate less conservative uncertainty-calibration and constraint-tightening mechanisms that retain the horizon-dependent error information of the present formulation while reducing conservatism and improving the ability to preserve nonempty tightened state sets as the prediction horizon increases. Another important direction is scaling the framework to multi-building and district-level DR, where the low online computational cost of explicit DPC policies may provide increasing advantages over repeated online optimization. Finally, for applications in which closed-loop stability or recursive-feasibility properties are essential, the horizon-indexed tightening structure developed here provides a natural basis for incorporating trajectory-shifting, invariance, terminal, or contraction-based arguments into future theoretical extensions.

\section*{CRediT authorship contribution statement}

Kaipeng Xu: Conceptualization, Methodology, Software, Validation, Formal analysis, Investigation, Data curation, Visualization, Writing -- original draft, Writing -- review \& editing, Project administration. Zhuo Zhi: Software, Writing -- review \& editing. Ruixuan Zhao: Writing -- review \& editing. Keyue Jiang: Writing -- review \& editing.

\section*{Declaration of competing interest}

The authors declare that they have no known competing financial interests or personal relationships that could have appeared to influence the work reported in this paper.

\section*{Data availability}

Data and code supporting the findings of this study will be made available on reasonable request.

\section*{Declaration of generative AI and AI-assisted technologies in the manuscript preparation process}

During the preparation of this work, the authors used OpenAI ChatGPT and Google Gemini for manuscript organization, language refinement, writing support, and assistance with code development and debugging. The authors reviewed and edited the output as needed and take full responsibility for the content of the published article.

\bibliographystyle{elsarticle-num}
\bibliography{literature}

\end{document}